\documentclass{article}
\usepackage[utf8]{inputenc}
\usepackage{a4wide}

\makeatletter
\expandafter\let\csname equation*\endcsname \@undefined{}
\expandafter\let\csname endequation*\endcsname \@undefined{}
\makeatother

\usepackage{amsmath}
\usepackage{amssymb}
\usepackage{mathpartir}
\usepackage{stmaryrd}
\usepackage{newtxmath}
\usepackage{ifthen}
\usepackage{xcolor}
\usepackage{epigraph}
\usepackage{tikz-cd}
\usepackage{float}
\usepackage{framed}
\usepackage{caption}
\usepackage{tcolorbox}
\usepackage{enumitem}
\floatstyle{plain}
\restylefloat{figure}
\definecolor{shadecolor}{rgb}{0.95,0.95,0.95}
\newcommand{\todo}[2][]{}
\newcommand{\code}[1]{\mathtt{#1}}
\definecolor{superlightgray}{gray}{0.80}

\newenvironment{ruledfigure}[1][!hbtp]{%
	\begin{figure}[#1]%
		\begin{center}%
		\noindent\rule{\textwidth}{1pt}%
		\small
}{%
	\noindent\rule{\textwidth}{1pt}%
	\end{center}%
	\end{figure}%
}

\newtcbox{\emphbox}{nobeforeafter,colframe=shadecolor,colback=shadecolor,boxrule=0pt,arc=4pt, boxsep=0pt,left=3pt,right=3pt,top=3pt,bottom=3pt,tcbox raise base}
\newtcbox{\emphunderlinebox}{nobeforeafter,colframe=shadecolor,colback=shadecolor,boxrule=0pt,arc=4pt, boxsep=0pt,left=3pt,right=3pt,top=3pt,bottom=1pt,tcbox raise base}

\newcommand{\defeq}{\stackrel{\text{def}}{=}}

\newcommand{\tyA}{\mathit{A}}
\newcommand{\tyB}{\mathit{B}}
\newcommand{\ctyC}{\underline{C}}
\newcommand{\ctyD}{\underline{D}}
\newcommand{\sig}{\Sigma}

\newcommand{\tmplT}{\mathit{T}}
\newcommand{\eqE}{\mathcal{E}}

\newcommand{\hcases}{\mathit{h}}
\newcommand{\contin}{\mathit{k}}
\newcommand{\continvar}{f}

\newcommand{\logic}{\ensuremath{\mathcal{L}}}
\newcommand{\fulllogic}{\ensuremath{\mathcal{L}_\text{full}}}
\newcommand{\freelogic}{\ensuremath{\mathcal{L}_\text{free}}}
\newcommand{\eqlogic}{\ensuremath{\mathcal{L}_\text{eq}}}
\newcommand{\emptylogic}{\ensuremath{\mathcal{L}_\emptyset}}
\newcommand{\predlogic}{\ensuremath{\mathcal{L}_\text{pred}}}

\newcommand{\unit}{\mathtt{(\,)}}
\newcommand{\true}{\mathtt{true}}
\newcommand{\false}{\mathtt{false}}

\newcommand{\keyin}{\mathtt{in}}
\newcommand{\keyreturn}{\mathtt{ret}}
\newcommand{\keydo}{\mathtt{do}}
\newcommand{\keyfun}{\mathtt{fun}}
\newcommand{\keyhandle}{\mathtt{handle}}
\newcommand{\keyhandler}{\mathtt{handler}}
\newcommand{\keywith}{\mathtt{with}}
\newcommand{\keyif}{\mathtt{if}}
\newcommand{\keythen}{\mathtt{then}}
\newcommand{\keyelse}{\mathtt{else}}

\newcommand{\op}{\mathit{op}}

\newcommand{\fun}[2]{\keyfun~{#1} \mapsto{#2}}
\newcommand{\apply}[2]{{#1}\;{#2}}
\newcommand{\return}[1]{\keyreturn~{#1}}
\newcommand{\dobind}[3]{\keydo~{#1} \leftarrow{#2}~\keyin~{#3}}
\newcommand{\betterdobind}[2]{\keydo~{#1} \leftarrow{#2}~\keyin~}
\newcommand{\ifclause}[3]{\keyif~{#1}~\keythen~{#2}~\keyelse~{#3}}
\newcommand{\operation}[3][\op]{
	\ifthenelse{\equal{#3}{}}
	{\ifthenelse{\equal{#2}{}}
		{#1}
		{#1 (#2)}}
	{#1 (#2;\,#3)}}
\newcommand{\withhandle}[2]{\keywith~{#1}~\keyhandle~{#2}}
\newcommand{\emptyhcases}{\nil}
\newcommand{\handler}[2][\return x \mapsto c_r]{\keyhandler \;({#1}; {#2})}
\newcommand{\subs}[2]{{#1}[{#2}]}
\newcommand{\subsfor}{/}

\newcommand{\rel}{\sim}
\newcommand{\teqjudgement}[2]{{#1}\vdash {#2}}
\newcommand{\tvars}{\mathtt{Z}}
\newcommand{\tenvlong}[2]{{#1};{#2}}
\newcommand{\tenv}{\tenvlong{\ctx}{\tvars}}
\newcommand{\thandle}[3][h]{{#2^#1} [{#3}]}
\newcommand{\respects}[5][\ctx]{#1 \vdash #2 \oftype #4 \casesto #5 \;\mathsf{respects}\; #3}
\newcommand{\bang}{!}
\newcommand{\witheq}{/}
\newcommand{\ctype}[3]{
	\ifthenelse{\equal{#3}{}}
	{#1\bang#2}
	{#1\bang#2\witheq#3}}
\newcommand{\ctypebasic}{\ctype{\tyA}{\sig}{\eqE}}
\newcommand{\set}[1]{\{{#1}\}}
\newcommand{\booltype}{\mathtt{bool}}
\newcommand{\inttype} {\mathtt{int}}
\newcommand{\unittype}{\mathtt{unit}}

\newcommand{\anytype}{\mathtt{*}}
\makeatletter
\newcommand{\superimpose}[2]{%
	{\ooalign{{\raise.2ex\hbox{$#1\@firstoftwo#2$}}\cr\hfil\raise-.2ex\hbox{$#1\@secondoftwo#2$}\hfil\cr}}}
\makeatother
\newcommand{\casesto}{\mathrel{\mathpalette\superimpose{{\rightharpoonup}{\rightharpoondown}}}}
\newcommand{\hto}{\Rightarrow}

\newcommand{\wfctx}[1]{\vdash #1 \oftype \mathsf{ctx}}
\newcommand{\wftctx}[1]{\vdash #1 \oftype \mathsf{tctx}}
\newcommand{\wfvtype}[1]{\vdash #1 \oftype \mathsf{vtype}}
\newcommand{\wfctype}[1]{\vdash #1 \oftype \mathsf{ctype}}
\newcommand{\wfsig}[1]{\vdash #1 \oftype \mathsf{sig}}
\newcommand{\wftemplate}[3][\tenv]{#1 \vdash #2 \oftype #3}
\newcommand{\wfequations}[2]{\vdash #1 \oftype #2}
\newcommand{\welltyped}[3][\ctx]{#1 \vdash #2 \oftype #3}
\newcommand{\oftype}{\!:\!}
\newcommand{\judgement}[3]{{#1} \vdash {#2} \oftype {#3}}

\newcommand{\subtypeof}{\leq}

\newcommand{\nil}{\emptyset}
\newcommand{\smallstep}[2]{{#1} \leadsto {#2}}

\newcommand{\lequal}{\equiv}
\newcommand{\eqjudg}[4][\ctx]{#1 \vdash #2 \lequal_{#4} #3}
\newcommand{\ljudgement}[3]{
\ifthenelse{\equal{#2}{}}
	{#1 \vdash #3}
	{{#1} \vsep {#2} \vdash {#3}}}
\newcommand{\hypo}{\Psi}
\newcommand{\ltrue}{\top}
\newcommand{\lfalse}{\bot}

\renewcommand{\land}{\wedge}
\renewcommand{\lor}{\vee}
\newcommand{\limplies}{\Rightarrow}
\newcommand{\lall}{\fall}
\newcommand{\lexists}{\exs}

\newcommand{\lift}[2][]{\mathsf{lift}_{#1} #2}
\newcommand{\interp}[2]{\mathsf{interp}_{#1} (#2)}
\newcommand{\freeInterp}[2]{F_{#1, #2}}
\newcommand{\inval}{\text{in}_{\keyreturn}}
\newcommand{\inop}[1][op]{\text{in}_{\mathit{#1}}}
\newcommand{\semTrue}{\mathrm{t\!t}}
\newcommand{\semFalse}{\mathrm{ff}}
\newcommand{\ctx}{\Gamma}

\newcommand{\semcases}{\tilde{h}}

\newcommand{\sem}[1]{\llbracket{#1} \rrbracket}
\newcommand{\jsem}[3]{\llbracket\judgement{#1}{#2}{#3} \rrbracket}
\newcommand{\hsem}[2][H]{\sem{#2}^{#1}}
\newcommand{\CEhole}[1]{\langle #1 \rangle}
\newcommand{\CEinstantiate}[2]{#1 \CEhole{#2}}
\newcommand{\CEvar}{\mathbb{C}}
\newcommand{\ctxequiv}[4][\ctx]{
	\judgement{#1}{#2 \cong #3}{#4}
}

\newcommand{\lambdafun}[3]{\lambda{#1} \in{#2} \,. \,{#3}}
\newcommand{\vsep}{\;|\;}
\newcommand{\union}{\cup}

\newcommand{\exs}[1]{\exists{#1}.\ }
\newcommand{\fall}[1]{\forall{#1}.\ }

\newcommand{\tab}{\hspace{3mm}}

\newenvironment{proof}
  {\emph{Proof.}}
  {\hfill$\square$}

\newtheorem{definition}{Definition}[section]
\newtheorem{example}[definition]{Example}
\newtheorem{lemma}[definition]{Lemma}
\newtheorem{theorem}[definition]{Theorem}
\newtheorem{proposition}[definition]{Proposition}
\newtheorem{corollary}[definition]{Corollary}

\begin{document}

\title{Local Algebraic Effect Theories}

\author{%
  {Žiga Lukšič\footnote{This material is based upon work supported by the Air Force Office of
  Scientific Research under award number FA9550-17-1-0326.}}\ \ and Matija Pretnar\footnotemark[1]\\
  University of Ljubljana, Faculty of Mathematics and Physics \\
  Slovenia
}

\date{\texttt{ziga.luksic@fmf.uni-lj.si}, \texttt{matija.pretnar@fmf.uni-lj.si}}

%===============================================================================

\maketitle

\begin{abstract}
  Algebraic effects are computational effects that can be described with a set of basic operations and equations between them. As many interesting effect handlers do not respect these equations, most approaches assume a trivial theory, sacrificing both reasoning power and safety.
  
  We present an alternative approach where the type system tracks equations that are observed in subparts of the program, yielding a sound and flexible logic, and paving a way for practical optimizations and reasoning tools.
\end{abstract}

%===============================================================================
% !TeX root = main.tex

\noindent
Algebraic effects are computational effects that can be described by a \emph{signature} of primitive operations and a collection of equations between them~\cite{DBLP:conf/fossacs/PlotkinP01, DBLP:journals/acs/PlotkinP03}, while algebraic effect \emph{handlers} are a generalization of exception handlers to arbitrary algebraic effects~\cite{DBLP:conf/esop/PlotkinP09, DBLP:journals/corr/PlotkinP13}. Even though the early work considered only handlers that respect equations of the effect theory, a considerable amount of useful handlers did not, and the restriction was dropped in most --- though not all~\cite{Ahman:PhDThesis, DBLP:journals/pacmpl/Ahman18} --- of the later work on handlers~\cite{DBLP:conf/icfp/KammarLO13, DBLP:journals/jlp/BauerP15, DBLP:conf/popl/Leijen17, DBLP:journals/pacmpl/BiernackiPPS18}, resulting in a weaker reasoning logic and imprecise specifications.

Our aim is to rectify this by reintroducing effect theories into the type system, tracking equations observed in parts of a program. On one hand, the induced logic allows us to rewrite computations into equivalent ones with respect to the effect theory, while on the other hand, the type system enforces that handlers preserve equivalences, further specifying their behaviour.
After an informal overview in Section~\ref{sec:overview}, we proceed as follows:
\begin{itemize}[topsep=5pt]
  \item
  The syntax of the working language, its operational semantics, and the typing rules are given in Section~\ref{sec:language}.
  \item
  Determining if a handler respects an effect theory is in general undecidable~\cite{DBLP:journals/corr/PlotkinP13}, so there is no canonical way of defining such a judgement. Therefore, the typing rules are given parametric to a reasoning logic, and in Section~\ref{sec:logics}, we present some of the more interesting choices.
  \item
  Since the definition of typing judgements is intertwined with a reasoning logic, we must be careful when defining the denotation of types and terms. Thus, in Section~\ref{sec:type-term-semantics}, we first introduce a set-based denotational semantics that disregards effect theories and prove the expected meta-theoretic properties.
  \item
  Next, in Section~\ref{sec:effect-theory-semantics}, we extend this denotation to templates and effect theories, and describe the necessary conditions for the reasoning logic which ensure its soundness and adequacy.
\end{itemize}
We conclude by discussing related and future work in Section~\ref{sec:conclusion}.

\section{Overview}\label{sec:overview}
% !TeX root = main.tex
% chktex-file 1
% chktex-file 46
%- - - - - - - - - - - - - - - - - - - - - - - - - - - - - - - - - - - - - - - -
\subsection{Algebraic effect handlers}

We assume the reader is vaguely familiar with algebraic effects and handlers, but we will elaborate on some of the more intricate parts. For less accustomed readers, a good place to pick up the basics is the tutorial~\cite{DBLP:journals/entcs/Pretnar15}. The emphasis will be on the newly introduced changes to the type and effect system.

At the core of using algebraic effects is the idea that all impure behaviour arises from calls of primitive operations, e.g.~$\mathit{print}$ for printing a string or $\mathit{raise}$ for raising an exception. Any computation either returns a value or makes an operation call~$\operation{v}{y.c}$ where the value~$v$ is a parameter of the call, while the computation~$c$ is its \emph{continuation}, waiting for a result of~$\op$ to be bound to~$y$.

Suppose we wish to model nondeterminism. For that, we take an operation $\mathit{choose} \oftype \unittype \to \booltype$ that non-deterministically produces a boolean after given the unit value~$\unit$. We can recover a binary non-deterministic choice from the abbreviation:
\[
  c_1 \oplus c_2
  \defeq
  \operation[\mathit{choose}]{\unit}{y. \ifclause{y}{c_1}{c_2}}
\]

Apart from a select few built-in operations (e.g.~printing out to a terminal), an operation by itself has no meaning. Their meaning is instead determined by a handler consisting of a set of operation clauses~$\operation{x}{k} \mapsto c_\op$. For every called operation $\op$, the handler replaces the call with the handling computation~$c_\op$, where the parameter of the call is bound to~$x$ and the continuation is captured in a function, which is recursively handled by the same handler, and bound to~$k$.

A simple example of a non-determinism handler is:
\begin{align*}
  &\mathit{pickLeft} \; = \; \keyhandler \; \{ \\
  &\tab | \; \mathit{choose} (\unit; k) \mapsto \apply{k}{\true}\\
  &\}
\end{align*}
that makes $\mathit{choose}$ constantly pass~$\true$ to the continuation~$k$, forcing $c_1 \oplus c_2$ to always pick $c_1$. Another example is the handler that collects all the results of a non-deterministic computation:
\begin{align*}
  &\mathit{collectToList} \; = \; \keyhandler \; \{ \\
  &\tab | \; \mathit{choose} (\unit; k) \mapsto \\
  &\tab\tab\tab \betterdobind{x_1}{\apply{k}{\true}} \\
  &\tab\tab\tab \betterdobind{x_2}{\apply{k}{\false}} \\
  &\tab\tab\tab \return{(x_1 @ x_2)}\\
  &\tab | \; \return{x} \mapsto \return{[x]} \\
  &\}
\end{align*}
If the handled computation calls the operation~$\mathit{choose}$, the handler passes both possible outcomes to the continuation, collects the respective results into lists $x_1$ and $x_2$, and returns the concatenated list. Additionally, the handler includes a clause stating that if the handled computation returns a value~$x$, the handler should transform it into a computation returning the singleton list~$[x]$.

\subsection{Effect theories}

Even though handlers determine the behaviour of operations, there are nonetheless some properties we expect from effects. These are described by a collection of equations called an \emph{effect theory}. For non-determinism, the theory consists of equations for commutativity, idempotency and associativity of the binary choice operation:
\begin{align}
  z_1 \oplus z_2 &\rel z_2 \oplus z_1, \tag{\textsc{comm}}\label{comm} \\
  z \oplus z &\rel z, \tag{\textsc{idem}}\label{idem} \\
  z_1 \oplus (z_2 \oplus z_3) &\rel (z_1 \oplus z_2) \oplus z_3 \tag{\textsc{assoc}}\label{assoc}
\end{align}
We quickly notice that $\mathit{pickLeft}$ does not respect the first equation by constructing a simple counter-example, for instance let $z_1 = \return{1}$ and $z_2 = \return{2}$. When handling the left side of the equation we obtain the result 1 (the left choice) while for the right side of the equation we get the result 2. Showing that $\mathit{pickLeft}$ respects the last two equations requires some additional tools and is done in Example~\ref{exa:pickLeft}. Similarly, the handler $\mathit{collectToList}$ respects the last equation but not the first two. The above two handlers are just two examples of many computationally interesting handlers that do not respect the usually assumed equations. For this reason, most contemporary work on algebraic effect handlers~\cite{DBLP:conf/icfp/KammarLO13, DBLP:journals/jlp/BauerP15, DBLP:conf/popl/Leijen17, DBLP:journals/pacmpl/BiernackiPPS18} assumes trivial effect theories that contain no equations.

Our proposed solution is to instead annotate computation types with equations that define the desired effect theory. For instance, consider the function
\[
  \mathit{chooseFromList} : A~\code{list} \to A \bang \set{\mathit{choose}}
\]
that takes a list of values of type~$A$ and non-deterministically chooses an element from it. The output type captures not only the type of returned values~$A$, but also the set~$\set{\mathit{choose}}$ of operations that may get called in the process. In our proposed system we further decorate the output type by stating the desired equations:
\[
  \mathit{chooseFromList} \oftype A~\code{list} \to A \bang \set{\mathit{choose}} \witheq \set{\eqref{comm},~\eqref{idem},~\eqref{assoc}}
\]
Now, by knowing only the type of~$\mathit{chooseFromList}$, its users can use induction~(see Example~\ref{exa:yield-proof1}) to show that computations
\[
  \dobind{x_1}{\apply{\mathit{chooseFromList}}{\ell_1}}{(\dobind{x_2}{\apply{\mathit{chooseFromList}}{\ell_2}}{c})}
\]
and
\[
  \dobind{x_2}{\apply{\mathit{chooseFromList}}{\ell_2}}{(\dobind{x_1}{\apply{\mathit{chooseFromList}}{\ell_1}}{c})}
\]
are equivalent at type $A \bang \set{\mathit{choose}} \witheq \set{\eqref{comm},~\eqref{idem},~\eqref{assoc}}$ (but not at $A \bang \set{\mathit{choose}} \witheq \nil$). Function implementers also benefit from the enriched types, as they can make additional assumptions on the observed behaviour. For example, by imposing the equation~\eqref{comm} on the output type, they ensure that the order in which they process the list is insignificant, as it will not be observable on the outside.

Any handler handling a computation with assumed equations must of course respect them. For example, $\mathit{collectToList}$ transforms a computation of type $A \bang \set{\mathit{choose}}$ into a pure computation of type $A~\code{list}$, but respects only~\eqref{assoc}. This is reflected in its type
\[
  \mathit{collectToList} \oftype A \bang \set{\mathit{choose}} \witheq \set{\eqref{assoc}} \hto A~\code{list} \bang \nil \witheq \nil
\]
in which the output computation calls no operations, which in turn leads to no possible equations between them. As the equations in the input type of $\mathit{collectToList}$ do not match those in the result of $\mathit{chooseFromList}$, the type system prohibits us from composing them. On the other hand, we could apply a handler $\mathit{collectToSet}$ that collects all results into a set and thus additionally respects~\eqref{comm} and~\eqref{idem}.

The proposed type system offers greater flexibility over assuming a global effect theory~\cite{DBLP:journals/corr/PlotkinP13}, as in one part of a program, we can assume arbitrary equations that the locally used handlers respect, but still use computationally interesting handlers that respect other equations in a different part.

Another flexibility that the type system offers is transforming one effect theory into another. For example, take an operation~$\mathit{yield} \oftype \inttype \to \unittype$ that is used to model integer generators. Using it, we design a handler that transforms a non-deterministic integer computation into a generator of all possible results:
\begin{align*}
  &\mathit{yieldAll} \; = \; \keyhandler \; \{ \\
  &\tab | \; \operation[\mathit{choose}]{\unit}{k} \mapsto \apply{k}{\true}; \apply{k}{\false} \\
  &\tab | \; \return{x} \mapsto \operation[\mathit{yield}]{x}{\_. \return{\unit}} \\
  &\}
\end{align*}
In the $\mathit{choose}$ clause, we first resume the continuation by passing it $\true$, yielding all outcomes in the process, and repeat the process for $\false$. Whenever a computation returns a value~$x$, we instead call $\mathit{yield}$ with a trivial continuation.

This handler respects none of the non-deterministic equations stated above, so its immediate type is:
\[
  \mathit{yieldAll} \oftype \inttype \bang \set{\mathit{choose}} \witheq \nil \hto \unittype \bang \set{\mathit{yield}} \witheq \nil.
\]
But if we assume that the order of $\mathit{yield}$ calls does not matter:
\begin{equation}
  \operation[\mathit{yield}]{x}{\_. \operation[\mathit{yield}]{y}{\_. c}}
  \rel
  \operation[\mathit{yield}]{y}{\_. \operation[\mathit{yield}]{x}{\_. c}}
  \tag{\textsc{yieldOrder}}
  \label{yieldOrder}
\end{equation}
the handler respects the commutativity of $\mathit{choose}$, and can be given a type
\[
  \mathit{yieldAll} \oftype \inttype \bang \set{\mathit{choose}} \witheq \set{\eqref{comm}} \hto \unittype \bang \set{\mathit{yield}} \witheq \set{\eqref{yieldOrder}},
\]
as we show in Example~\ref{exa:yield-proof1}. If we next have a handler that respects~\eqref{yieldOrder}, for example
\[
  \mathit{sumYieldedValues} \oftype
  \unittype \bang \set{\mathit{yield}} \witheq \set{\eqref{yieldOrder}}
  \hto
  \inttype \bang \nil \witheq \nil,
\]
and a computation~$c$ in which we assume~\eqref{comm}, we can compose them as
\[
  \withhandle{\mathit{sumYieldedValues}}{(\withhandle{\mathit{yieldAll}}{c})},
\]
with the type system ensuring that equations are preserved at appropriate places. Again, we can use the equation~\eqref{comm} and rewrite~$c$ into an equivalent computation, all by knowing only the type and not the exact definition of handlers above it.

While the equations look simple, they are expressive enough to state even more intricate properties. Consider integer generators with an operation $\mathit{next} \oftype \unittype \to \inttype \; \code{ option}$. A call of the operation $\mathit{next}$ either generates the next element $\code{Some}\,n$ of the generator, or returns $\code{None}$ if the generator has finished generating the sequence. We expect that after such a call, all further calls should result in $\code{None}$. We can specify such behaviour with the equation
\[
  \operation[\mathit{next}]{\unit}{y.\ifclause{(y == \code{None})}{\operation[\mathit{next}]{\unit}{y'.\apply{z}{y'}}}{\apply{z}{\code{None}}}}
  \rel 
  \operation[\mathit{next}]{\unit}{y.\apply{z}{\code{None}}}
\]
where the variable $z$ stands for an arbitrary computation, dependent on a value of type $\inttype \; \code{ option}$. To see that this equation describes the desired behaviour, consider any ill-behaved handler that passes $\code{None}$ to the continuation on the first call of $\mathit{next}$ and $\code{Some}\,n$ on the second one. Handling the left hand side first passes $\code{None}$ to $y$, leading to the second call of $\mathit{next}$, which is handled by passing $\code{Some}\,n$ to $y'$ and further on to $z$. Handling the right-hand side, however, resumes by passing $\code{None}$ to $z$, resulting in a different computation. In any other case (if the handler passes $\code{Some}\,n$ in the first call or $\code{None}$ in the second one), both sides proceed by handling $\apply{z}{\code{None}}$.

\section{Language}\label{sec:language}
% !TeX root = main.tex
% chktex-file 1
% chktex-file 46

\subsection{Term Syntax}

Our working calculus (Figure~\ref{fig:term-syntax}) is based on the fine-grain call-by-value~\cite{DBLP:journals/iandc/LevyPT03} approach, which differentiates between pure \emph{values}~$v$ and effectful \emph{computations}~$c$, which might return a value or call an operation. For clarity, we keep the calculus minimal, though it could easily be extended with additional value types such as integers, type sums and products, or recursive function definitions, which we discuss in Section~\ref{sub:future-work}.

\begin{ruledfigure}
	% !TeX root = ../main.tex
% chktex-file 1
% chktex-file 46

\[
	\begin{array}{rrll}
		\text{values}~v
		 & ::=   & x                                                                              & \text{variable}          \\
		 & \vsep & \unit                                                                          & \text{unit constant}     \\
		 & \vsep & \true \vsep \false                                                             & \text{boolean constants} \\
		 & \vsep & \fun{x}{c}                                                                     & \text{function}          \\
		 & \vsep & \handler{\hcases}                                                              & \text{handler}           \\
		 \\
		 \text{computations}~c
		 & ::=   & \ifclause{v}{c_1}{c_2}                                                         & \text{conditional}       \\
		 & \vsep & \apply{v_1}{v_2}                                                               & \text{application}       \\
		 & \vsep & \return{v}                                                                     & \text{returned value}    \\
		 & \vsep & \operation{v}{y.c}                                                             & \text{operation call}    \\
		 & \vsep & \dobind{x}{c_1}{c_2}                                                           & \text{sequencing}        \\
		 & \vsep & \withhandle{v}{c}                                                              & \text{handling}					\\
		 \\
		 \text{operation clauses}~\hcases
		  & ::=   & \emptyhcases \vsep \hcases \union \set{\operation{x}{\contin} \mapsto c_\op}                            \\
		\end{array}
		\]

	\caption{Syntax of terms.}\label{fig:term-syntax}
\end{ruledfigure}

\noindent
For a cleaner development later on, we define an independent syntactic sort of \emph{operation clauses}~$h$ which are joint with a \emph{return clause} only when constructing the handler value. This deviates slightly from most of the contemporary work on handlers~\cite{DBLP:conf/popl/Leijen17, DBLP:conf/icfp/HillerstromL16,DBLP:journals/pacmpl/0002KLP17, DBLP:conf/esop/SalehKPS18} and is more similar to the original treatment in~\cite{DBLP:conf/esop/PlotkinP09, DBLP:journals/corr/PlotkinP13}. In practice, we treat operation clauses~$h$ as a set of operations with uniquely assigned handling computations and write them as $\set{\operation{x}{\contin} \mapsto c_\op}_\op$.

To improve readability we sometimes use syntactic sugar. When sequencing, we replace computations of form $\dobind{\_}{c_1}{c_2}$ with $c_1; c_2$. When writing handler clauses we often use the separator $\vsep$ instead of commas, to achieve the familiar ML pattern matching look.

%- - - - - - - - - - - - - - - - - - - - - - - - - - - - - - - - - - - - - - - -
\subsection{Operational Semantics}

Operational semantics, given in Figure~\ref{fig:operational}, remains largely the same as in our previous work~\cite{DBLP:journals/jlp/BauerP15, DBLP:journals/jfp/KammarP17} except for a different presentation of operation clauses. Computations continue to evaluate until they either return a value or call an operation. At that point we propagate the operation by pushing the remaining parts of the computation inside the operation continuation. This way we ensure that the operation eventually reaches a handler, while at the same time correcting the continuation.

In line with our previous work~\cite{DBLP:journals/jlp/BauerP15, DBLP:journals/jfp/KammarP17}, the presented handlers are \emph{deep}~\cite{DBLP:conf/icfp/KammarLO13}, meaning that they continue to handle any operations called in the potentially resumed continuation. For the sake of a slightly simpler type system, we choose to switch to \emph{closed handlers}~\cite{DBLP:conf/icfp/KammarLO13}, which get stuck on operation calls with no corresponding operation clauses. Our type system will prevent such cases, though it is straightforward to extend the semantics to \emph{open} handlers, where unhandled operations are implicitly propagated outwards, or add propagating cases $\operation{x}{\contin} \mapsto \operation{x}{y. \apply{\contin}{y}}$ that do that explicitly.

\begin{ruledfigure}
	% !TeX root = ../main.tex
% chktex-file 1
% chktex-file 46

%=-=-=-=-=-=-=-=-=-=-=-=-=-=-=-=-=-=-=-=-=-=-=-=-=-=-=-=-=-=-=-=-=-=-=-=-=-=-=-=
\begin{mathpar}
	
	% IF CLAUSE
	\inferrule*[]
	{
	}
	{ \smallstep
		{\ifclause{\true}{c_1}{c_2}}
		{c_1}
	}
	
	\inferrule*[]
	{
	}
	{ \smallstep
		{\ifclause{\false}{c_1}{c_2}}
		{c_2}
	}
	
	% FUNCTION
	\inferrule*[]
	{
	}
	{ \smallstep
		{\apply{(\fun{x}{c})}{v}}
		{\subs{c}{v\subsfor x}}
	}

	% DO BIND
	\inferrule*[]
	{\smallstep{c_1}{c_1'}
	}
	{ \smallstep
		{\dobind{x}{c_1}{c_2}}
		{\dobind{x}{c_1'}{c_2}}
	}
	
	\inferrule*[]
	{
	}
	{ \smallstep
		{\dobind{x}{\return{v}}{c}}
		{\subs{c}{v\subsfor x}}
	}
	
	\inferrule*[]
	{
	}
	{ \smallstep
		{\dobind{x}{\operation{v}{y.c_1}}{c_2}}
		{\operation{v}{y.\dobind{x}{c_1}{c_2}}}
	}
	
	% HANDLE
	
	\inferrule*[]
	{\smallstep{c}{c'}
	}
	{ \smallstep
		{\withhandle{v}{c}}
		{\withhandle{v}{c'}}
	}
	
	\inferrule*[]
	{
	}
	{ \smallstep
		{\withhandle{(\handler\hcases)}{(\return v)}}
		{\subs{c_r}{v \subsfor x}}
	}
	
	\mprset{flushleft}
	\inferrule*[]
	{
		(\operation{x}{\contin} \mapsto c_\op) \in \hcases
	}
	{ \withhandle{(\handler\hcases)}{(\operation{v}{y.c})} \\
	  \smallstep{}{\subs{c_\op}{v \subsfor x, (\fun{y}{\withhandle{(\handler\hcases)}{c}}) \subsfor k}}
	}
	
\end{mathpar}

	\caption{Operational semantics.}\label{fig:operational}
\end{ruledfigure}	

%- - - - - - - - - - - - - - - - - - - - - - - - - - - - - - - - - - - - - - - -
\subsection{Type Syntax}

Continuing the separation between values and computations, the type syntax (Figure~\ref{fig:type-syntax}) distinguishes between \emph{value types}~$\tyA, \tyB$, often referred to simply as types, and \emph{computation types} $\ctyC, \ctyD$. In addition to the type of the returned value $\tyA$, a computation type~$\ctype{\tyA}{\sig}{\eqE}$ captures the signature~$\sig$ of operations that might be called during evaluation and, the main novelty of this paper, a set of equations, called a \emph{(effect) theory}~$\eqE$, which describes which computations we consider equivalent at a given type.

\begin{ruledfigure}
	% !TeX root = ../main.tex
% chktex-file 1
% chktex-file 46

%=-=-=-=-=-=-=-=-=-=-=-=-=-=-=-=-=-=-=-=-=-=-=-=-=-=-=-=-=-=-=-=-=-=-=-=-=-=-=-=
\[
	\begin{array}{rrll}
		\text{(value) type}~\tyA, \tyB
		 & ::=   & \unittype          & \text{unit type}\\
		 & \vsep & \booltype          & \text{boolean type}\\
		 & \vsep & \tyA \to \ctyC     & \text{function type}\\
		 & \vsep & \ctyC \hto \ctyD   & \text{handler type}\\
		\\
		\text{computation type}~\ctyC, \ctyD
		 & ::=   & \ctype{\tyA}{\sig}{\eqE}\\
		\\
		\text{signature}~\sig
		 & ::=   & \nil \vsep \sig \union \set{{\op \oftype \tyA \to \tyB}}\\
		\\
		\text{value context}~\ctx
		 & ::=   & \varepsilon \vsep \ctx, x \oftype \tyA \\
		\\
		\text{template context}~\tvars
		 & ::=   & \varepsilon \vsep \tvars, z \oftype \tyA \to \anytype \\
		\\
		\text{template}~\tmplT
		 & ::=   & \apply{z}{v} & \text{applied template variable}\\
		 & \vsep & \ifclause{v}{\tmplT_1}{\tmplT_2} & \text{conditional template}\\
		 & \vsep & \operation{v}{y.\tmplT} & \text{operation call template}\\
		\\
		\text{(effect) theory}~\eqE
		 & ::=   & \nil \vsep \eqE \union
		\set{\teqjudgement{\tenv}{\tmplT_1 \rel \tmplT_2}}\\
	\end{array}
\]

	\caption{Syntax of types.}\label{fig:type-syntax}
\end{ruledfigure}

\noindent
Operation calls exhibit answer-type polymorphism in the sense that they prescribe only the type that the continuation expects, but not its return type. Equations between operations have to be similarly polymorphic, so we describe them with a pair of \emph{templates}~$\tmplT$~\cite{DBLP:journals/corr/PlotkinP13}. A template pair is then instantiated to a pair of computations by replacing all template variables with appropriate function values. In particular, given a function value $f_j$ for each free template variable $z_j$ appearing in a template~$\tmplT$, we recursively define the computation $\subs{\tmplT}{f_j \subsfor z_j}_j$ by:
\begin{align*}
	\subs{(\apply{z_j}{v})}{f_j \subsfor z_j}_j
	&= \apply{f_j}{v} \\
	\subs{(\ifclause{v}{\tmplT_1}{\tmplT_2})}{f_j \subsfor z_j}_j
	&= \ifclause{v}{\subs{\tmplT_1}{f_j \subsfor z_j}_j}{\subs{\tmplT_2}{f_j \subsfor z_j}_j} \\
	\subs{(\operation{v}{y.\tmplT})}{f_j \subsfor z_j}_j
	&= \operation{v}{y.\subs{\tmplT}{f_j \subsfor z_j}_j}
\end{align*}
Since templates have to be polymorphic in the result type, we have to restrict their constructors to a small answer-type polymorphic subset, though one which proves to be sufficient for many interesting and common effect theories.

\subsection{Type Checking}\label{sub:type-rules}

It should come as no surprise that the proposed type system is intricate. Types contain templates, which then further contain terms and types, and to top it all off, both terms and types are split into mutually dependent value and computation sorts. This forces us to mutually define judgements for:
\begin{itemize}
	\item $\welltyped{v}{\tyA}$, which states that in context $\ctx$ the value $v$ has a type~$\tyA$,
	\item $\welltyped{c}{\ctyC}$, which states that in context $\ctx$ the computation $c$ has a computation type~$\ctyC$,
	\item $\welltyped{\hcases}{\sig \casesto \ctyD}$, which states that in context $\ctx$ clauses $\hcases$ cover operations listed in $\sig$ using computations of type $\ctyD$,
	\item $\respects{\hcases}{\eqE}{\sig}{\ctyD}$, which states that in context $\ctx$ and with respect to the signature $\sig$, clauses $\hcases$ are well-defined, meaning they handle computations equivalent under~$\eqE$ into equivalent computations of type $\ctyD$,
	\item $\wfvtype{\tyA}$, which states that the value type $\tyA$ is well-formed,
	\item $\wfctype{\ctyC}$, which states that the computation type $\ctyC$ is well-formed,
	\item $\wfsig{\sig}$, which states that the signature $\sig$ is well-formed,
	\item $\wfctx{\ctx}$, which states that the value context $\ctx$ is well-formed,
	\item $\wftctx{\tvars}$, which states that the template context $\tvars$ is well-formed,
	\item $\wftemplate{\tmplT}{\sig}$, which states that in contexts $\ctx$ and $\tvars$, the template $\tmplT$ is well-formed with respect to the signature $\sig$,
	\item $\wfequations{\eqE}{\sig}$, which states that equations $\eqE$ are well-formed with respect to the signature~$\sig$.
\end{itemize}	
Even though the forthcoming rules have to be treated as a single definition, we structure them into smaller, more digestible chunks.

\begin{ruledfigure}
	% !TeX root = ../main.tex
% chktex-file 1
% chktex-file 46

\subsubsection*{Well-typed values~\emphbox{$\welltyped{v}{\tyA}$} (where $\wfctx{\ctx}$ and $\wfvtype{\tyA}$)}

\begin{mathpar}

	% Term Variable
	\inferrule*[]
	{ (x \oftype \tyA) \in \ctx
	}
	{ \judgement{\ctx}{x}{\tyA}
	}

	% Constant unit
	\inferrule*[]
	{ }
	{ \judgement{\ctx}{\unit}{\unittype}
	}

	% Boolean
	\inferrule*[right=]
	{ }
	{ \judgement{\ctx}{\true}{\booltype}
	}

	\inferrule*[]
	{ }
	{ \judgement{\ctx}{\false}{\booltype}
	}

	% Function
	\inferrule*[]
	{ \judgement{\ctx, x \oftype \tyA}{c}{\ctyC}
	}
	{ \judgement{\ctx}{\fun{x}{c}}{\tyA \to \ctyC}
	}

	% Handler
	\inferrule*[]
	{
		\judgement{\ctx, x \oftype \tyA}{c_r}{\ctyD}
		\\
		\respects{\hcases}{\eqE}{\sig}{\ctyD}
		\\
	}
	{\judgement{\ctx}
		{\handler{\hcases}}
		{\ctype{\tyA}{\sig}{\eqE} \hto \ctyD}
	}
\end{mathpar}

%=-=-=-=-=-=-=-=-=-=-=-=-=-=-=-=-=-=-=-=-=-=-=-=-=-=-=-=-=-=-=-=-=-=-=-=-=-=-=-=
\subsubsection*{Well-typed computations~\emphunderlinebox{$\welltyped{c}{\ctyC}$} (where $\wfctx{\ctx}$ and $\wfctype{\ctyC}$)}

\begin{mathpar}
	% Conditional
	\inferrule*[]
	{ \judgement{\ctx}{v}{\booltype}\\
		\judgement{\ctx}{c_1}{\ctyC}\\
		\judgement{\ctx}{c_2}{\ctyC}
	}
	{ \judgement{\ctx}{\ifclause{v}{c_1}{c_2}}{\ctyC}
	}

	% Function application
	\inferrule*[]
	{ \judgement{\ctx}{v_1}{\tyA \to \ctyC}\\
		\judgement{\ctx}{v_2}{\tyA}
	}
	{ \judgement{\ctx}{\apply{v_1}{v_2}}{\ctyC}
	}

	% Return
	\inferrule*[]
	{ \judgement{\ctx}{v}{\tyA}
	}
	{ \judgement{\ctx}{\return{v}}{\ctype{\tyA}{\sig}{\eqE}}
	}

	% Operation
	\inferrule*[]
	{ (\op \oftype \tyA_\op \to \tyB_\op) \in\sig \\
		\judgement{\ctx}{v}{\tyA_\op}\\
		\judgement{\ctx, y \oftype \tyB_\op}{c}{\ctype{\tyA}{\sig}{\eqE}}
	}
	{ \judgement{\ctx}{\operation{v}{y.c}}{\ctype{\tyA}{\sig}{\eqE}}
	}

	% Do-bind
	\inferrule*[]
	{ \judgement{\ctx}{c_1}{\ctype{\tyA}{\sig}{\eqE}} \\
		\judgement{\ctx, x \oftype \tyA}{c_2}{\ctype{\tyB}{\sig}{\eqE}} \\
	}
	{ \judgement{\ctx}{\dobind{x}{c_1}{c_2}}{\ctype{\tyB}{\sig}{\eqE}}
	}

	% Handling
	\inferrule*[]
	{ \judgement{\ctx}{v}{\ctyC \hto \ctyD}\\
		\judgement{\ctx}{c}{\ctyC}
	}
	{ \judgement{\ctx}{\withhandle{v}{c}}{\ctyD}
	}

\end{mathpar}

	\caption{Typing judgements for terms.}\label{fig:typing-judgements}
\end{ruledfigure}

\noindent
First, Figure~\ref{fig:typing-judgements} lists the usual typing rules for values and computations. Aside from the decoupling of operation clauses, most of the rules closely follow our previous work~\cite{DBLP:journals/corr/BauerP13, DBLP:journals/entcs/Pretnar15, DBLP:journals/jfp/KammarP17}. The main difference is the addition of equations, which are in all but one case simply tacked onto the computation type. For example in sequencing~$\dobind{x}{c_1}{c_2}$, we require that $c_1$ and $c_2$ have not only matching signatures, but also equations.

The rule for typing handlers is more interesting. If a handler is given a type~$\ctypebasic \hto \ctyD$, we must first check that the return clause maps values of type~$\tyA$ to computations of type~$\ctyD$. Next, all operations in $\sig$ must have appropriate operation clauses that handle them with computations of type~$\ctyD$, well-defined with respect to equations~$\eqE$. All of this is captured by the judgement $\respects{\hcases}{\eqE}{\sig}{\ctyD}$, given in Figure~\ref{fig:operation-clauses-judgements}. There are different choices one can consider in defining this judgement, though all impose the same typing rules, captured by the auxiliary judgement~$\judgement{\ctx}{\hcases}{\sig \casesto {\ctyD}}$. We present a few interesting choices of resulting logics~$\logic$ in Section~\ref{sec:logics}.

\begin{ruledfigure}
	
%=-=-=-=-=-=-=-=-=-=-=-=-=-=-=-=-=-=-=-=-=-=-=-=-=-=-=-=-=-=-=-=-=-=-=-=-=-=-=-=
\subsubsection*{Well-typed operation clauses~\emphunderlinebox{$\welltyped{\hcases}{\sig \casesto \ctyD}$} (where $\wfctx{\ctx}$, $\wfsig{\sig}$ and $\wfctype{\ctyD}$)}

\begin{mathpar}
	% Operation clauses

	\inferrule*[]{	}
	{\judgement{\ctx}
		{\emptyhcases}
		{\nil \casesto {\ctyD}}
	}

	\inferrule*[]
	{\judgement{\ctx}
		{\hcases}
		{\sig \casesto {\ctyD}}\\
		\judgement{\ctx, x \oftype \tyA_\op, \contin \oftype \tyB_\op \to \ctyD}
		{c_\op}
		{\ctyD}\\
		\op \not \in \sig
	}
	{\judgement{\ctx}
		{\hcases \union \set{\operation{x}{\contin} \mapsto c_\op}}
		{(\sig \union \set{\op \oftype \tyA_\op \to \tyB_\op}) \casesto {\ctyD}}
	}

\end{mathpar}

%=-=-=-=-=-=-=-=-=-=-=-=-=-=-=-=-=-=-=-=-=-=-=-=-=-=-=-=-=-=-=-=-=-=-=-=-=-=-=-=
\subsubsection*{Well-defined operation clauses~\emphunderlinebox{$\respects{\hcases}{\eqE}{\sig}{\ctyD}$} (where $\wfequations{\eqE}{\sig}$ and $\welltyped{\hcases}{\sig \casesto \ctyD}$)}

\begin{center}
	Given in Section~\ref{sec:logics}.
\end{center}

	\caption{Typing judgements for operation clauses.}\label{fig:operation-clauses-judgements}
\end{ruledfigure}

\noindent
Rules that ensure well-formedness of types, signatures and contexts are given in Figure~\ref{fig:wellformed} and are routine. The most interesting one is the rule for the computation type $\ctypebasic$, which requires the theory~$\eqE$ to be well-formed with respect to the signature~$\sig$.

We treat templates and equations in Figure~\ref{fig:typing-templates}. The rules for templates follow the corresponding rules for computations, while equations are well-formed with respect to the signature~$\sig$ if all their templates are. Template variables in template contexts~$\tvars$ are labelled with the type $\tyA \to \anytype$ as they can be replaced with functions of type $\tyA \to \ctyC$ for an arbitrary computation type~$\ctyC$.

\begin{theorem}[Safety]\label{theorem:safety}
	\emph{Progress.}
	If $\judgement{}{c}{\ctypebasic}$ then either
	\begin{itemize}[noitemsep, topsep=0pt]
		\item there exists a computation $c'$ such that $\smallstep{c}{c'}$, or
		\item $c$ is of the form $\return v$ for some value $v$, or
		\item $c$ is of the form $\operation{v}{k}$ for some $\op \in \sig$.
	\end{itemize}
	\emph{Preservation.}
	If $\judgement{}{c}{\ctypebasic}$ and $\smallstep{c}{c'}$ then $\judgement{}{c'}{\ctypebasic}$.
\end{theorem}
\begin{proof}
	Both parts can be shown with a routine structural induction.
\end{proof}

\begin{ruledfigure}
	% !TeX root = ../main.tex
% chktex-file 1
% chktex-file 46

%=-=-=-=-=-=-=-=-=-=-=-=-=-=-=-=-=-=-=-=-=-=-=-=-=-=-=-=-=-=-=-=-=-=-=-=-=-=-=-=
\subsubsection*{Well-formed value types~\emphbox{$\wfvtype{\tyA}$}}

\begin{mathpar}
	\inferrule*[]{
	}{
		\wfvtype{\unittype}
	}

	\inferrule*[]{
	}{
		\wfvtype{\booltype}
	}

	\inferrule*[]{
		\wfvtype{\tyA} \\
		\wfctype{\ctyC}
	}{
		\wfvtype{\tyA \to \ctyC}
	}

	\inferrule*[]{
		\wfctype{\ctyC} \\
		\wfctype{\ctyD}
	}{
		\wfvtype{\ctyC \hto \ctyD}
	}
\end{mathpar}

%=-=-=-=-=-=-=-=-=-=-=-=-=-=-=-=-=-=-=-=-=-=-=-=-=-=-=-=-=-=-=-=-=-=-=-=-=-=-=-=
\subsubsection*{Well-formed computation types~\emphunderlinebox{$\wfctype{\ctyC}$}}

\begin{mathpar}
	\inferrule*[]{
		\wfvtype{\tyA} \\
		\wfsig{\sig} \\
		\wfequations{\eqE}{\sig}
	}{
		\wfctype{\ctypebasic}
	}
\end{mathpar}

%=-=-=-=-=-=-=-=-=-=-=-=-=-=-=-=-=-=-=-=-=-=-=-=-=-=-=-=-=-=-=-=-=-=-=-=-=-=-=-=
\subsubsection*{Well-formed signatures~\emphbox{$\wfsig{\sig}$}}

\begin{mathpar}
	\inferrule*[]{
	}{
		\wfsig{\nil}
	}

	\inferrule*[]{
		\wfsig{\sig} \\
		\wfvtype{\tyA} \\
		\wfvtype{\tyB} \\
		\op \not\in \sig
	}{
		\wfsig{\sig \union \set{\op \oftype \tyA \to \tyB}}
	}
\end{mathpar}

%=-=-=-=-=-=-=-=-=-=-=-=-=-=-=-=-=-=-=-=-=-=-=-=-=-=-=-=-=-=-=-=-=-=-=-=-=-=-=-=
\subsubsection*{Well-formed context~\emphbox{$\wfctx{\ctx}$} and template context~\emphbox{$\wftctx{\tvars}$}}

\begin{mathpar}
	\inferrule*[]{
	}{
		\wfctx{\nil}
	}

	\inferrule*[]{
		\wfctx{\ctx} \\
		x \not\in \ctx \\
		\wfvtype{\tyA}
	}{
		\wfctx{\ctx, x \oftype \tyA}
	}

	\inferrule*[]{
	}{
		\wftctx{\nil}
	}

	\inferrule*[]{
		\wftctx{\tvars} \\
		z \not\in \tvars \\
		\wfvtype{\tyA}
	}{
		\wftctx{\tvars, z \oftype \tyA \to \anytype}
	}
\end{mathpar}

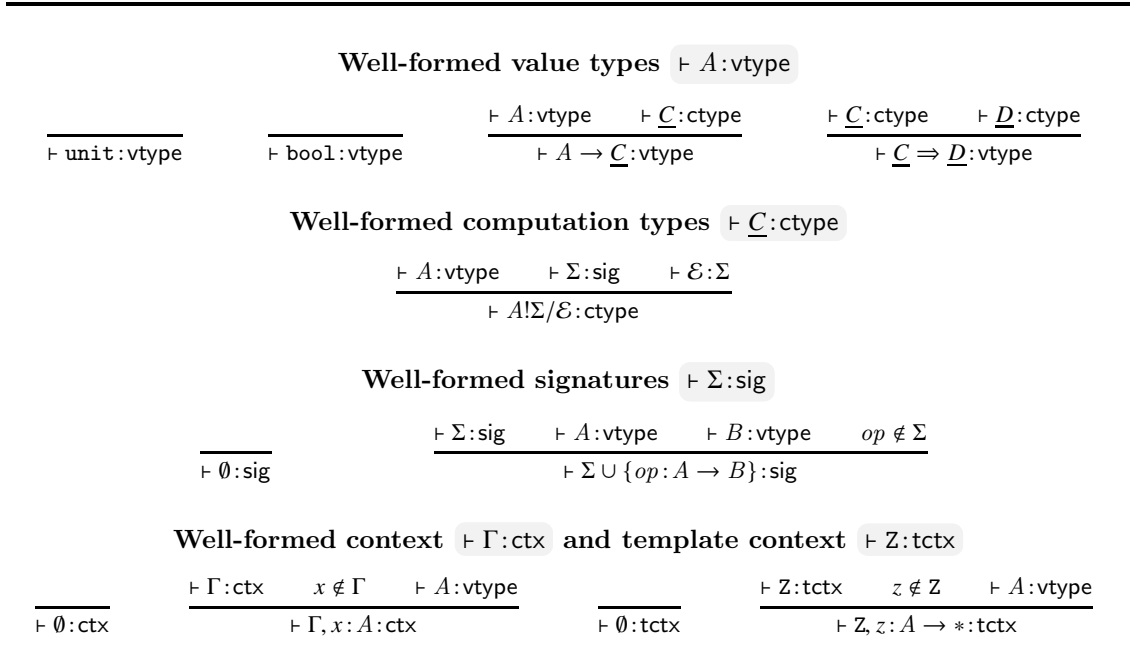
	\caption{Well-formedness judgements for types.}\label{fig:wellformed}
\end{ruledfigure}	

\begin{ruledfigure}
	% !TeX root = ../main.tex
% chktex-file 1
% chktex-file 46

%=-=-=-=-=-=-=-=-=-=-=-=-=-=-=-=-=-=-=-=-=-=-=-=-=-=-=-=-=-=-=-=-=-=-=-=-=-=-=-=
\subsubsection*{Well-typed templates~\emphbox{$\wftemplate{\tmplT}{\sig}$} (where $\wfctx{\ctx}$, $\wftctx{\tvars}$, and $\wfsig{\sig}$)}

\begin{mathpar}

	% Term Variable
	\inferrule*[]
	{ \welltyped{v}{\tyA} \\
	  (z \oftype \tyA \to \anytype) \in \tvars
	}
	{ \wftemplate{\apply{z}{v}}{\sig}
	}

	% Conditional
	\inferrule*[]
	{ \welltyped{v}{\booltype}\\
	  \wftemplate{\tmplT_1}{\sig} \\
	  \wftemplate{\tmplT_2}{\sig}
	}
	{ \wftemplate{\ifclause{v}{\tmplT_1}{\tmplT_2}}{\sig}
	}

	% Sequencing
	\inferrule*[]
	{ (\op \oftype \tyA \to \tyB) \in \sig \\
		\welltyped[\ctx]{v}{\tyA} \\
		\wftemplate[\ctx, y \oftype \tyB;\tvars]{\tmplT}{\sig}
	}
	{ \wftemplate{\operation{v}{y.\tmplT}}{\sig}
	}
\end{mathpar}

%=-=-=-=-=-=-=-=-=-=-=-=-=-=-=-=-=-=-=-=-=-=-=-=-=-=-=-=-=-=-=-=-=-=-=-=-=-=-=-=
\subsubsection*{Well-formed theories~\emphbox{$\wfequations{\eqE}{\sig}$} (where $\wfsig{\sig}$)}

\begin{mathpar}
	\inferrule*[]{
	}{
		\wfequations{\nil}{\sig}
	}

	\inferrule*[]{
		\wfequations{\eqE}{\sig} \\
		\wftemplate{\tmplT_1}{\sig} \\
		\wftemplate{\tmplT_2}{\sig}
	}{
		\wfequations{\eqE \union \set{\teqjudgement{\tenv}{\tmplT_1 \rel \tmplT_2}}}{\sig}
	}
\end{mathpar}

	\caption{Typing judgements for templates.}\label{fig:typing-templates}
\end{ruledfigure}

\section{Logics}\label{sec:logics}
% !TeX root = main.tex
% chktex-file 1
% chktex-file 46

We now consider different possibilities for the definition of $\respects{\hcases}{\eqE}{\sig}{\ctyD}$. The trivial choice is to take the empty relation, resulting in a logic~\emptylogic\ in which no handlers can be typed. At the other extreme, we can take the logic~\fulllogic\ with the unwieldy set of all possible combinations of $\welltyped{\hcases}{\sig \casesto \ctyD}$ and $\wfequations{\eqE}{\sig}$ that yield a well-defined denotation~(cf.~Definition~\ref{def:sound-logic}). In between, there are a few interesting logics on which we focus.

\subsection{Free logic~\freelogic}

The simplest one of these is the logic~\freelogic\ in which well-typed operation clauses respect only the empty set of equations.
\[
	\inferrule*[]
	{ \welltyped{\hcases}{\sig \casesto \ctyD} }
	{ \respects{\hcases}{\nil}{\sig}{\ctyD} }
\]
This corresponds to the conventional approach to handlers in which we ignore equations and accept any well-typed handler. Note that we could replace $\nil$ with an arbitrary set of tautologies such as $\teqjudgement{\tenv}{\tmplT \rel \tmplT}$. The same principle applies in general, as we may replace any set of equations with an equivalent one.

\subsection{Equational logic~\eqlogic}\label{sub:equational-logic}

We next consider a simple equational logic~\eqlogic{}, which allows us to prove that operation clauses indeed respect a given theory. In addition to the type rules, which are the same as in Section~\ref{sub:type-rules}, \eqlogic{} includes the $\respects{\hcases}{\eqE}{\sig}{\ctyD}$ relation and typed equational judgements $\eqjudg{v_1}{v_2}{\tyA}$ for values and $\eqjudg{c_1}{c_2}{\ctyC}$ for computations. Most of these additional rules are well-known: reflexivity, symmetry, transitivity, substitution, congruences for each construct, and $\beta\eta$-equivalences. More interesting are the rules given in Figure~\ref{fig:equational_logic}.

\begin{ruledfigure}
	% !TeX root = ../main.tex
% chktex-file 1
% chktex-file 46
%-----------------------------------------------------------------------------
\subsubsection*{Inheriting equations from the effect theory}

\begin{mathpar}

	\inferrule*[]
		{ \left(\teqjudgement{\tenvlong
				{(x_i \oftype \tyA_i)_i}
				{(z_j \oftype \tyB_j \to \anytype)_j}
			}{\tmplT_1 \rel \tmplT_2}
		  \right) \in \eqE
		  \\
		  \welltyped{v_i}{\tyA_i}
		  \\
		  \welltyped{f_j}{\tyB_j \to \ctypebasic}
		}
		{ \eqjudg
			{\subs{(\subs{T_1}{f_j \subsfor z_j}_j)}{v_i \subsfor x_i}_i}
			{\subs{(\subs{T_2}{f_j \subsfor z_j}_j)}{v_i \subsfor x_i}_i}
			{\ctypebasic}
		}
																	
\end{mathpar}

%-----------------------------------------------------------------------------
\subsubsection*{Checking when a handler respects an equational theory}

\begin{mathpar}

	\inferrule*[]
	{ \welltyped{\hcases}{\sig \casesto \ctyD} }
	{ \respects{\hcases}{\nil}{\sig}{\ctyD} }
	\\

	% Respect
	\inferrule*[]
	{ 
		\respects{\hcases}{\eqE}{\sig}{\ctyD} \\
		\ljudgement{\ctx, (x_i \oftype \tyA_i)_i, (f_j \oftype \tyB_j \to \ctyD)_j}{}
		{\thandle{T_1}{f_j \subsfor z_j}_j 
		\lequal_{\ctyD} 
		\thandle{T_2}{f_j \subsfor z_j}_j}
	}
	{ \respects{\hcases}{\;\eqE \union \left \{
		\teqjudgement{\tenvlong
			{(x_i \oftype \tyA_i)_i}
			{(z_j \oftype \tyB_j \to \anytype)_j}}
		{\tmplT_1 \rel \tmplT_2}
		\right \} \;}{\sig}{\ctyD}
	}
	
\end{mathpar}

\vspace{3mm}
where for $h = \set{\operation{x}{k} \mapsto c_\op}_\op$ we define:
\begin{align*}
  \thandle{z_i(v)}{f_j \subsfor z_j}_j &= \apply{f_i}{v} \\
  \thandle{(\ifclause{v}{\tmplT_1}{\tmplT_2})}{f_j \subsfor z_j}_j &= 
  \ifclause{v}
  {\thandle{\tmplT_1}{f_j \subsfor z_j}_j}
  {\thandle{\tmplT_2}{f_j \subsfor z_j}_j}\\
  \thandle{\operation{v}{y.\tmplT}}{f_j \subsfor z_j}_j &=
  \subs{c_\op}{v \subsfor x, (\fun{y}{\thandle{\tmplT}{f_j \subsfor z_j}_j}) \subsfor k}
\end{align*}

	\caption{Non-standard rules of the logic~\eqlogic.}\label{fig:equational_logic}
\end{ruledfigure}

The first rule allows us to use equations we have in computation types. Any instantiation of templates $\tmplT_1 \rel \tmplT_2$ in $\eqE$ produces an equivalence between computations at any type $\ctypebasic$. This rule is a generalisation of the rule present in the original logic for algebraic effects~\cite{DBLP:conf/lics/PlotkinP08}, except that the effect theory is local rather than global.

When instantiating the two templates in an equation, we need to ensure that the final terms are of equal types. This is ensured by the following lemma.

\begin{lemma}\label{lem:template_inst}
Suppose we have a well-typed template $\wftemplate[\tenvlong{{(x_i \oftype \tyA_i)}_i}{{(z_j \oftype \tyB_j \to\anytype)}_j}]{\tmplT}{\sig}$ and values $\welltyped{v_i}{\tyA_i}$ for each $i$ and $\welltyped{f_j}{\tyB_j \to \ctyC}$ for each $j$. Then, we have ${\welltyped{\subs{(\subs{\tmplT}{f_j \subsfor z_j}_j)}{v_i \subsfor x_i}_i}{\ctyC}}$.
\end{lemma}

The last two rules describe when a handler respects an effect theory and are similar to the rules present in the original treatment of handlers~\cite{DBLP:conf/esop/PlotkinP09, DBLP:journals/corr/PlotkinP13}, but adapted to local effect theories. The first of the two rules treats the empty theory as before and the second allows us to extend the theory with a single equation between templates. Lemma~\ref{lem:template_inst} again guarantees that equations in the hypotheses are well typed.

To show how rules of \eqlogic\ can be used in practice, let us return to the running examples from Section~\ref{sec:overview}.

\begin{example}\label{exa:pickLeft}
	Recall the definition of the handler
	\begin{align*}
		&\mathit{pickLeft} \; = \; \keyhandler \; \{ \\
		&\tab | \; \mathit{choose} (\unit; k) \mapsto \apply{k}{\true}\\
		&\}
	\end{align*}
	We wish to show that it has the type:
	\[
	  \mathit{pickLeft} \oftype \ctype{\inttype}{\set{\mathit{choose}}}{\set{\eqref{idem},\eqref{assoc}}} \hto \ctype{\inttype}{\emptyset}{\emptyset}
	\]
	We first focus on the equation~{\eqref{idem}}:
	\[
		z \oftype \unittype \to \anytype \vdash
		\operation[\mathit{choose}]{\unit}{y. \ifclause{y}{\apply{z}{\unit}}{\apply{z}{\unit}}} \rel {\apply{z}{\unit}}
	\]
	%
	% and
	% %
	% \begin{multline*}
	% 	z_1 \oftype \unittype \to \anytype,
	% 	z_2 \oftype \unittype \to \anytype;
	% 	z_3 \oftype \unittype \to \anytype \vdash \\
	% 	\operation[\mathit{choose}]{\unit}{y. 
	% 		\ifclause{y}
	% 			{\apply{z_1}{\unit}}
	% 			{\operation[\mathit{choose}]{\unit}{y. 
	% 				\ifclause{y}{\apply{z_2}{\unit}}{\apply{z_3}{\unit}}
	% 				}
	% 			}
	% 		}
	% 	\\
	% 	\rel
	% 	\\
	%   \operation[\mathit{choose}]{\unit}{y. 
	% 		\ifclause{y}
	% 			{\operation[\mathit{choose}]{\unit}{y. 
	% 				\ifclause{y}{\apply{z_2}{\unit}}{\apply{z_3}{\unit}}
	% 				}
	% 			}
	% 			{\apply{z_1}{\unit}}
	% 		}
	% \end{multline*}
	% %
	To show that $\mathit{pickLeft}$ respects~\eqref{idem}, we must prove that for any $f \oftype \unittype \to \ctype{\inttype}{\emptyset}{\emptyset}$ (recall this type is obtained by instantiating $\anytype$ with the right-hand side of the handler type) it holds that
	\[
	\thandle[\mathit{pickLeft}]{(\operation[\mathit{choose}]{\unit}{y. \ifclause{y}{\apply{z}{\unit}}{\apply{z}{\unit}}})}{f \subsfor z}
	\lequal_{\ctype{\inttype}{\emptyset}{\emptyset}}
	\thandle[\mathit{yieldAll}]{({\apply{z}{\unit}})}{f \subsfor z}.
	\]
	We rewrite both sides according to definition of $\thandle[\mathit{pickLeft}]{(\_)}{f \subsfor z}$ to the formula
	\[
		\apply{(\fun{y}{\ifclause{y}{\apply{f}{\unit}}{\apply{f}{\unit}}})}{\true}
		\lequal_{\ctype{\inttype}{\emptyset}{\emptyset}}
		\apply{f}{\unit}
	\]
	By using $\beta$-laws the left side first simplifies to
	\[
		{\ifclause{\true}{\apply{f}{\unit}}{\apply{f}{\unit}}}
	\]
	and then further to
	\(
		\apply{f}{\unit}
	\)
	which concludes the proof.

	The proof for \eqref{assoc}, while requiring more space to write out, is no more difficult. For the left side we rewrite $\thandle[\mathit{pickLeft}]{(z_1 \oplus (z_2 \oplus z_3))}{f_1 \subsfor z_1, f_2 \subsfor z_2, f_3 \subsfor z_3}$ to
	\[
		(\fun{y}{\ifclause{y}{\apply{f_1}{\unit}}{(\fun{y'}{\ifclause{y'}{\apply{f_2}{\unit}}{\apply{f_3}{\unit}}}) \; \true}}) \; \true
	\]
	which can be simplified to $\apply{f_1}{\unit}$. The same process is repeated for the right side of the equation.
\end{example}

\begin{example}\label{exa:yield-proof1}
	Recall the definition of the handler
	\begin{align*}
	  &\mathit{yieldAll} \; = \; \keyhandler \; \{ \\
	  &\tab | \; \operation[\mathit{choose}]{\unit}{k} \mapsto \apply{k}{\true}; \apply{k}{\false} \\
	  &\tab | \; \return{x} \mapsto \operation[\mathit{yield}]{x}{\_. \return{\unit}} \\
	  &\}
	\end{align*}
	We wish to show that it can be given the type:
	\[
	  \mathit{yieldAll} \oftype \ctype{\inttype}{\set{\mathit{choose}}}{\set{\eqref{comm}}} \hto \ctype{\unittype}{\set{\mathit{yield}}}{\set{\eqref{yieldOrder}}},
	\]
	where the equations~{\eqref{comm}} and~\eqref{yieldOrder} are
	\begin{multline*}
		z_1 \oftype \unittype \to \anytype,
		z_2 \oftype \unittype \to \anytype \vdash \\
		\operation[\mathit{choose}]{\unit}{y. \ifclause{y}{\apply{z_1}{\unit}}{\apply{z_2}{\unit}}} \rel \operation[\mathit{choose}]{\unit}{y. \ifclause{y}{\apply{z_2}{\unit}}{\apply{z_1}{\unit}}}
	\end{multline*}
	and
	\[
		x \oftype \inttype,
		y \oftype \inttype;
		z \oftype \unittype \to \anytype \vdash
		\operation[\mathit{yield}]{x}{\_. \operation[\mathit{yield}]{y}{\_. \apply{z}{\unit}}}
	  \rel
	  \operation[\mathit{yield}]{y}{\_. \operation[\mathit{yield}]{x}{\_. \apply{z}{\unit}}}	
	\]
	respectively.
	
	To show that $\mathit{yieldAll}$ respects~\eqref{comm}, it is enough that we prove that for any functions $f_1, f_2 \oftype \unittype \to \ctype{\unittype}{\set{\mathit{yield}}}{\set{\eqref{yieldOrder}}}$ it holds that
	\begin{multline*}
	\thandle[\mathit{yieldAll}]{(\operation[\mathit{choose}]{\unit}{y. \ifclause{y}{\apply{z_1}{\unit}}{\apply{z_2}{\unit}}})}{f_1 \subsfor z_1, f_2 \subsfor z_2}
	\\
	\lequal_{\ctype{\unittype}{\set{\mathit{yield}}}{\set{\eqref{yieldOrder}}}}
	\\
	\thandle[\mathit{yieldAll}]{(\operation[\mathit{choose}]{\unit}{y. \ifclause{y}{\apply{z_2}{\unit}}{\apply{z_1}{\unit}}})}{f_1 \subsfor z_1, f_2 \subsfor z_2}.
	\end{multline*}
	
	By using the definition of $\mathit{yieldAll}$, the left hand-side can be rewritten to the sequence of computations
	\begin{align*}
	&\apply{(\fun{y}{\ifclause{y}{\apply{f_1}{\unit}}{\apply{f_2}{\unit}}})}{\true};\\
	&\apply{(\fun{y}{\ifclause{y}{\apply{f_1}{\unit}}{\apply{f_2}{\unit}}})}{\false}.
	\end{align*}
	which is $\beta$-equivalent to $\apply{f_1}{\unit}; \apply{f_2}{\unit}$. We repeat the process for the right-hand side to obtain the equation:
	\[
	\apply{f_1}{\unit}; \apply{f_2}{\unit}
	\lequal_{\ctype{\unittype}{\set{\mathit{yield}}}{\set{\eqref{yieldOrder}}}} 
	\apply{f_2}{\unit}; \apply{f_1}{\unit}.
	\]
	At this step, we postpone the remainder of the proof to Example~\ref{exa:yield-proof2}, as \eqlogic\ is unfortunately not powerful enough to finish it. A crucial piece missing is the principle of computational induction~\cite{DBLP:conf/lics/PlotkinP08, DBLP:journals/corr/BauerP13}, which captures the inductive structure of computation types $\ctypebasic$.
\end{example}
%- - - - - - - - - - - - - - - - - - - - - - - - - - - - - - - - - - - - - - - -

\subsection{Predicate logic with induction~\predlogic}\label{sub:full-logic}

In order to state induction in our logic, we need to extend our judgements with hypotheses and universal quantifiers. We now extend the logic to a first-order predicate logic~\predlogic. In addition to equations, the \emph{formulae}~$\varphi$ include logical connectives and quantifiers over value types:
\[
	\begin{array}{rrll}
		\text{formulae}~\varphi, \psi
		 & ::=   & v_1 \lequal_{\tyA} v_2                                                                             & \text{value equation}          \\
		 & \vsep & c_1 \lequal_{\ctyC} c_2                                                                          & \text{computation equation}     \\
		 & \vsep & \ltrue                                                                          & \text{truth}     \\
		 & \vsep & \lfalse                                                                          & \text{falsity}     \\
		 & \vsep & \varphi_1 \land \varphi_2                                                             & \text{conjunction} \\
		 & \vsep & \varphi_1 \lor \varphi_2                                                             & \text{disjunction} \\
		 & \vsep & \varphi \limplies \psi                                                             & \text{implication} \\
		 & \vsep & \lall{x \oftype \tyA} \varphi                                                                    & \text{universal quantification}          \\
		 & \vsep & \lexists{x \oftype \tyA} \varphi                                                                    & \text{existential quantification} \\
	\end{array}
\]
We type formulae in a context $\ctx$ and extend judgements with hypotheses to ones of the form $\ljudgement{\ctx}{\hypo}{\varphi}$, where $\hypo$ is a set of formulae $\psi_1, \dots, \psi_n$. In addition to the rules of \eqlogic{} (extended with hypotheses), the logic~\predlogic{} includes the standard rules for logical connectives and quantifiers, which we omit (cf.~\cite{DBLP:conf/lics/PlotkinP08, DBLP:phd/ethos/Pretnar10}), and a principle of induction, on which we focus now.

The principle of induction states that a property holds for all computations of type $\ctypebasic$, if it holds for all computations that return a value of type $\tyA$, and for all computations that call an operation~$\op \in \sig$ under the induction hypotheses that it holds for all possible continuations. For a schema producing a formula $\varphi(c)$ for any computation $c$, the induction is stated as: 
\[
	\inferrule*[]
	{
	\ljudgement{\ctx}{\hypo}{c \oftype \ctypebasic}
	\\
	\ljudgement{\ctx, x \oftype \tyA}{\hypo}{\varphi(\return x)}
	\\
	\Big[\ljudgement{\ctx, x \oftype \tyA_\op, k \oftype \tyB_\op \to \ctypebasic}
		{\hypo, \big(\fall{y \oftype \tyB_\op} \varphi(\apply{k}{y})\big)}{\varphi(\operation{x}{y.\apply{k}{y}})}\Big]
	_{\op : \tyA_\op \to \tyB_\op \in \sig}
	}
	{ \ljudgement{\ctx}{\hypo}{\varphi(c)}
	}
\]

\begin{example}\label{exa:yield-proof2}
	Using induction, we may finally complete the proof started in Example~\ref{exa:yield-proof1}. Recall we were left at proving
	\[
		\apply{f_1}{\unit}; \apply{f_2}{\unit}
		\lequal_{\ctyD} 
		\apply{f_2}{\unit}; \apply{f_1}{\unit}.
	\]
	where we abbreviate $\ctyD = \ctype{\unittype}{\set{\mathit{yield}}}{\set{\eqref{yieldOrder}}}$.
	We wish to show with induction that in $\ctyD$ any two computations commute. To prove this, we first show that a single call of $\mathit{yield}$ commutes with any computation in $\ctyD$, so we take~$\varphi_1(c_1)$ to be: 
	\[
		 \mathit{yield}(x; \_.c_1); c_2 \lequal_{\ctyD} c_1; \mathit{yield}(x; \_.c_2) .
	\]
	In the proof we include hints on what rule we used (for instance ``$\beta$ for $;$ and $\op$'' means we used the $\beta$-law dealing with sequencing and operations). We first prove the base case for $c_1 = \return \unit$:
	\begin{align*}
		&\mathit{yield}(x; \_.\return \unit);\; c_2 &\\
		{}    \lequal_{\ctyD} {}& \mathit{yield}(x; \_.(\return \unit;\; c_2)) & (\text{$\beta$ for $;$ and $\op$})\\
		{}    \lequal_{\ctyD} {}& \mathit{yield}(x; \_.c_2) & (\text{$\beta$ for $;$ and $\return{}$}) \\
		{}    \lequal_{\ctyD} {}& \return \unit;\; \mathit{yield}(x; \_.c_2) & (\text{$\beta$ for $;$ and $\return{}$, other direction})
	\end{align*}

	Next, we take $c_1 = \mathit{yield}(y; \_. \apply{k}{\unit})$ and prove the induction step using the hypothesis
	\[
		\mathit{yield}(x; \_.\apply{k}{\unit}); c_2 \lequal_{\ctyD} \apply{k}{\unit}; \mathit{yield}(x;\_.c_2).
	\]
	We proceed as:
	\begin{align*}
		&\mathit{yield}(x; \_.\mathit{yield}(y; \_.\apply{k}{\unit}));\; c_2 &\\
		{} \lequal_{\ctyD} {}&\mathit{yield}(y; \_.\mathit{yield}(x; \_.\apply{k}{\unit}));\; c_2 & (\text{\eqref{yieldOrder} holds in $\ctyD$})\\
		{} \lequal_{\ctyD} {}& \mathit{yield}(y; \_.\mathit{yield}(x; \_.\apply{k}{\unit});\; c_2) & (\text{$\beta$ for $;$ and $\op$})\\
		{} \lequal_{\ctyD} {}& \mathit{yield}(y; \_.\apply{k}{\unit};\; \mathit{yield}(x; \_.c_2)) & (\text{induction hypothesis})\\
		{} \lequal_{\ctyD} {}& \mathit{yield}(y; \_.\apply{k}{\unit}) ;\; \mathit{yield}(x; \_.c_2) & (\text{$\beta$ for $;$ and $\op$, other direction})
	\end{align*}
	%Here we used~\eqref{yieldOrder} to obtain the first equality and in the third line we used the induction hypothesis.

	We now show that any two computations in $\ctyD$ commute. For that we take $\varphi_2(c_1)$ to be:
	\[
		c_1; c_2 \lequal_{\ctyD} c_2; c_1.
	\]
	We again first show the base case for $c_1 = \return \unit$ by using the $\beta$-equivalence for sequencing and return:
	\[
		\return \unit;\; c_2 \lequal_{\ctyD} c_2 \lequal_{\ctyD} c_2;\; \return \unit
	\]
	We then show the induction step for $c_1 = \mathit{yield}(y; \_. \apply{k}{\unit})$ with the hypothesis $\apply{k}{\unit}; c_2 \lequal_{\ctyD} c_2; \apply{k}{\unit}$:
	\begin{align*}
		&\mathit{yield}(x; \_.\apply{k}{\unit}); \; c_2 &\\
		\lequal_{\ctyD}\;& \mathit{yield}(x; \_.\apply{k}{\unit}; \; c_2) & (\text{$\beta$ for $;$ and $\op$})\\
		\lequal_{\ctyD}\;& \mathit{yield}(x; \_.c_2; \; \apply{k}{\unit}) & (\text{induction hypothesis})\\
		\lequal_{\ctyD}\;& \mathit{yield}(x; \_.c_2); \; \apply{k}{\unit} & (\text{$\beta$ for $;$ and $\op$, other direction})\\
		\lequal_{\ctyD}\;& c_2; \; \mathit{yield}(x; \_.\apply{k}{\unit}) & (\text{$\varphi_1(\apply{k}{\unit})$ from previous proof})
	\end{align*}
	%where we have used $\varphi_1(k \unit)$ in the last step.
\end{example}

\section{Denotation of types and terms}\label{sec:type-term-semantics}
% !TeX root = main.tex
% chktex-file 1
% chktex-file 46
% chktex-file 3

We have presented a type system in which well-formed types depend on well-typed terms and vice-versa. To avoid circularity when defining the denotational semantics of such types and terms, we proceed in two stages. First, we define a denotation of types that is independent of effect theories. This allows us to further define the denotation of well-typed terms, which similarly does not take effect theories into an account. In Section~\ref{sec:effect-theory-semantics} we then equip each type with an equivalence relation that stems from the effect theory and show that the term denotations are well-defined with respect to it.

%===============================================================================
\subsection{Semantics of types}\label{sub:semantics-types}
%===============================================================================

To value types $\tyA$ (and computation types $\ctyC$) we assign sets $\sem{\tyA}$ (and $\sem{\ctyC}$) as follows:
\begin{align*}
	\sem{\unittype}        & = \set{\star}                  &
	\sem{\booltype}        & = \set{\semFalse, \semTrue}      \\
	\sem{\tyA \to \ctyC}   & = \sem{\tyA} \to {\sem{\ctyC}} &
	\sem{\ctyC \hto \ctyD} & = \sem{\ctyC} \to \sem{\ctyD}    \\
	\sem{\ctypebasic}      & = \sem{\sig}\sem{\tyA}
\end{align*}
where for $\sig = {\{ \op \oftype \tyA_\op \to \tyB_\op \}}_\op$, we define $\sem{\sig}$ to be the free functor mapping a set $X$ to the inductively defined set $\sem{\sig} X$ containing:
\begin{enumerate}
	\item $\inval(a)$ for each $a$ in $X$
	\item $\inop(a; \kappa)$ for each $\op \oftype \tyA_\op \to \tyB_\op \in \sig$, each $a \in \sem{\tyA_\op}$ and each $\kappa \in \sem{\tyB_\op} \to \sem{\sig} X$
\end{enumerate}
Note that handlers are interpreted by ordinary functions and $\eqE$ does not play a role in the denotation of $\ctypebasic$.

Next, an \emph{interpretation}~$H$ of a signature~$\sig$ over a set $Y$ is a family of functions $H_\op : \sem{\tyA_\op} \times (\sem{\tyB_\op} \to Y) \to Y$ for each $\op \; \oftype \; \tyA_\op \to \tyB_\op \in \sig$. We define the set~$\interp{\sig}{Y}$ of all interpretations by:
\[
	\interp{\sig}{Y} =
	\prod_{\op \; \oftype \; \tyA_\op \to \tyB_\op \in \sig}
	\sem{\tyA_\op} \times (\sem{\tyB_\op} \to Y) \to Y
\]
For any signature~$\sig$ and set~$X$, we define a \emph{free interpretation}~$\freeInterp{X}{\sig} \in \interp{\sig}{\sem{\sig} X}$ by:
\[
	(\freeInterp{X}{\sig})_\op(a)(\kappa) = \inop(a; \kappa)
\]
Next, for any interpretation $H \oftype \interp{\sig}{Y}$, we can \emph{lift} a function $f \oftype X \to Y$ to a function $\lift[H]{f} \oftype \sem{\sig}X \to Y$, defined recursively by:
\begin{align*}
	\lift[H]{f}(\inval(x))        & = f(x),\\
	\lift[H]{f}(\inop(x; \kappa)) & = H_\op(x; \lift[H]{f} \circ \kappa).
\end{align*}

%===============================================================================
\subsection{Semantics of well-typed values and computations}\label{sub:semant-expr-comp}
%===============================================================================

Well-typed terms
\begin{align*}
	\judgement{\ctx}{v}{\tyA}
	\qquad\text{and}\qquad
	\judgement{\ctx}{c}{\ctyC}
\end{align*}
are interpreted as maps
\begin{equation*}
	\sem{\judgement{\ctx}{v}{\tyA}} : \sem{\ctx} \to \sem{\tyA}
	\qquad\text{and}\qquad
	\sem{\judgement{\ctx}{c}{\ctyC}} : \sem{\ctx} \to \sem{\ctyC},
\end{equation*}
where $\ctx$ is defined component-wise:
\begin{align*}
	\sem{\varepsilon}          & = \set{\star}                  \\
	\sem{\ctx, x \oftype \tyA} & = \sem{\ctx} \times \sem{\tyA}
\end{align*}
The definition proceeds by recursion on the derivation of the typing judgment. When no confusion can arise, we abbreviate the denotations to $\sem{v}$ and $\sem{c}$.

Given an environment $\eta \in \sem{\ctx}$, the rules for base values are
\begin{align*}
	\jsem{\ctx}{x_i}{A_i}\eta          & = \eta_i    &
	\jsem{\ctx}{\unit}{\unittype}\eta  & = \star       \\
	\jsem{\ctx}{\false}{\booltype}\eta & = \semFalse &
	\jsem{\ctx}{\true}{\booltype}\eta  & = \semTrue
\end{align*}
while for functions, we have:
\begin{equation*}
	\jsem{\ctx}{(\fun{x}{c})}{\tyA \to \ctyC}\eta =
	\lambdafun{a}{\sem{\tyA}}{\jsem{\ctx, x \oftype \tyA}{c}{\ctyC}{(\eta, a)}}
\end{equation*}

In order to define the denotation of handlers, we must first treat operation clauses. A set of well-typed clauses~$\welltyped{\hcases}{\sig \casesto \ctyD}$ is defined as a map
\begin{equation*}
	\sem{\judgement{\ctx}{\hcases}{\sig \casesto \ctyD}} : \sem{\ctx} \to \interp{\sig}{\sem{\ctyD}}
\end{equation*}
where
\begin{multline*}
	\sem{\judgement{\ctx}{\set{\operation{x}{\contin} \mapsto c_\op}_\op}{\sig \casesto \ctyD}}\eta = \\
	{\big \{
		\lambdafun{a}{\sem{\tyA_\op}}{
			\lambdafun{\kappa}{\sem{\tyB_\op \to \ctyD}}{
				\jsem{\ctx, x \oftype \tyA_\op, \contin \oftype \tyB_\op \to \ctyD}{c_\op}{\ctyD}(\eta, a, \kappa)
			}}
		\big \}}_{\op \; \oftype \; \tyA_\op \to \tyB_\op \in \sig}	
\end{multline*}

A denotation of a handler is just the lifting of its return clause to the interpretation given by the operation clauses:
\begin{equation*}
	\jsem{\ctx}{\handler \hcases}{\ctypebasic \hto \ctyD}\eta = \lift[\sem{\hcases}\eta]{(\lambdafun{a}{\sem{\tyA}}{\sem{c_r}(\eta, a)})}
\end{equation*}
In contrast to the original denotational semantics of handlers~\cite{DBLP:journals/corr/PlotkinP13}, effect theories do not affect the denotation of types, so all handlers receive a denotation.

The denotation of computations is more or less structural:
\begin{align*}
	\jsem{\ctx}{\ifclause{v}{c_1}{c_2}}{\ctyC}\eta & =
	\begin{cases}
		\jsem{\ctx}{c_1}{\ctyC}\eta & \text{if $\jsem{\ctx}{v}{\booltype}\eta = \semTrue$,} \\
		\jsem{\ctx}{c_2}{\ctyC}\eta & \text{if $\jsem{\ctx}{v}{\booltype}\eta = \semFalse$}
	\end{cases} \\
	\jsem{\ctx}{\apply{v_1}{v_2}}{\ctyC}\eta       & =
		(\jsem{\ctx}{v_1}{\tyA \to \ctyC}\eta) (\jsem{\ctx}{v_2}{\tyA}\eta)
\end{align*}
Returned values and operations of type $\ctypebasic$ are interpreted by appropriate constructors of $\sem{\sig}{\sem{\tyA}}$:
\begin{align*}
	\jsem{\ctx}{\return v}{\ctypebasic}\eta 
		& = \inval(\jsem{\ctx}{v}{\tyA}\eta) \\
	\jsem{\ctx}{\operation{v}{y.c}}{\ctypebasic}\eta 
		& = \\
		& \hspace{-20mm}\inop (\jsem{\ctx}{v}{\tyA_\op}\eta; \lambdafun{b}{\sem{\tyB_\op}}{\jsem{\ctx, y \oftype \tyB_\op}{c}{\ctypebasic}{(\eta,b)}}).
\end{align*}
The denotation of sequencing is obtained by lifting the continuation of the second sequent to the free interpretation and applying it to the first sequent:
\begin{multline*}
	\jsem{\ctx}{\dobind{x}{c_1}{c_2}}{\ctype{\tyB}{\sig}{\eqE}}\eta = \\
	\lift[\freeInterp{\sem{\tyA}}{\sig}]{(\lambdafun{a}{\sem{\tyA}}{\jsem{\ctx, x \oftype \tyA}{c_2}{\ctype{\tyB}{\sig}{\eqE}}{(\eta, a)}})}
	(\jsem{\ctx}{c_1}{\ctypebasic}\eta).
\end{multline*}
Finally, since handlers are ordinary functions, handling is just application:
\begin{align*}
	\jsem{\ctx}{\withhandle{v}{c}}{\ctyD}\eta      & =
		(\jsem{\ctx}{v}{\ctyC \hto \ctyD}\eta) (\jsem{\ctx}{c}{\ctyC}\eta)
\end{align*}

%===============================================================================
\subsection{Relation to operational semantics}\label{sub:semantics-properties}
%===============================================================================

Before we proceed, we first ensure that the presented denotational semantics is sound with respect to the operational semantics.

\begin{proposition}\label{prop:soundness-operational}
	If $\welltyped[]{c}{\ctyC}$ and $\smallstep{c}{c'}$ then $\welltyped[]{c'}{\ctyC}$ and $\jsem{}{c}{\ctyC} = \jsem{}{c'}{\ctyC}$.
\end{proposition}
\begin{proof}
	The fact that $c'$ has the same type follows from type preservation in Theorem~\ref{theorem:safety}, while the proof of the second part proceeds by an easy induction on the derivation of $\smallstep{c}{c'}$.
\end{proof}

As expected, our denotational semantics is also adequate with respect to the operational semantics.

\begin{lemma}[Adequacy]\label{lem:adequacy}
	If $\jsem{}{c}{\ctype{\booltype}{\nil}{\nil}} = \inval{(\jsem{}{v}{\booltype})}$ then $c \leadsto^{*} \return{v}$.
\end{lemma}
\begin{proof}
	As our language features no recursion, Theorem~\ref{theorem:safety} implies that there is some sequence of steps $c \leadsto^{*} \return{v'}$. From Proposition~\ref{prop:soundness-operational} it follows that $\sem{c} = \inval{(\sem{v'})}$, thus $\sem{v} = \sem{v'}$. Because distinct boolean values receive distinct denotations, we have $v = v'$ so $c \leadsto^{*} \return{v}$.
\end{proof}

If our language featured recursion, we would instead follow the adequacy proof presented in~\cite{DBLP:journals/corr/BauerP13}, as our operational and denotational semantics are a simplification (we omit instances and subtyping) of the ones given there.

We define a \emph{computation context}~$\CEvar$ as a computation with a number of holes~$\CEhole{\,}$ (possibly under binders) into which one may plug a computation~$c$ to obtain a computation $\CEinstantiate{\CEvar}{c}$. A context~$\CEvar$ is a \emph{ground computation context for $\ctx$ and $\ctyC$} if for all computations $\welltyped{c}{\ctyC}$, we have $\welltyped[]{\CEinstantiate{\CEvar}{c}}{\ctype{\booltype}{\nil}{\nil}}$.

Computations $\welltyped{c}{\ctyC}$ and $\welltyped{c'}{\ctyC}$ are \emph{contextually equivalent}, written as $\ctxequiv{c}{c'}{\ctyC}$ if for all ground computation contexts $\CEvar$ for $\ctx$ and $\ctyC$ we have $\CEinstantiate{\CEvar}{c} \leadsto^* \return{\true}$ if and only if $\CEinstantiate{\CEvar}{c'} \leadsto^* \return{\true}$. We similarly define $\ctxequiv{v}{v'}{\tyA}$ for values.

\begin{corollary}\label{cor:eq_adequacy}
	If $\jsem{\ctx}{c}{\ctyC} = \jsem{\ctx}{c'}{\ctyC}$, then $\ctxequiv{c}{c'}{\ctyC}$.
	If $\jsem{\ctx}{v}{\tyA} = \jsem{\ctx}{v'}{\tyA}$, then $\ctxequiv{v}{v'}{\tyA}$.
\end{corollary}
\begin{proof}
	The proof is a folklore application of adequacy. Assume that $\CEinstantiate{\CEvar}{c} \leadsto^* \return{\true}$. By Proposition~\ref{prop:soundness-operational}, $\sem{\CEinstantiate{\CEvar}{c}} = \inval(\semTrue)$. However, denotational semantics is structural, thus $\sem{\CEinstantiate{\CEvar}{c}} = \sem{\CEinstantiate{\CEvar}{c'}}$, which by Lemma~\ref{lem:adequacy} implies $\CEinstantiate{\CEvar}{c'} \leadsto^* \return{\true}$. The proof for values is identical.
\end{proof}

\section{Denotation of effect theories and logics}\label{sec:effect-theory-semantics}
% !TeX root = main.tex
% chktex-file 1
% chktex-file 46
% chktex-file 3

Having provided a sound denotational semantics that disregards effect theories, we proceed by equipping it with relations that reflect the theories.

%===============================================================================
\subsection{Semantics of templates}\label{sub:semantics-templates}
%===============================================================================

We first turn our attention to templates, which are the basic building blocks of equations. As template variable contexts~$\tvars$ are polymorphic in the type of computations, we interpret them as functors~$\sem{\tvars}$, defined by:
\begin{align*}
	\sem{\varepsilon} Y & = \set{\star} \\
	\sem{\tvars, z \oftype \tyA \to \anytype} Y & = \sem{\tvars}Y \times Y^{\sem{\tyA}}
\end{align*}
We overload the notation and write $\sem{\ctx}$ for the constant functor that maps any set $Y$ to the set $\sem{\ctx}$ as defined in Section~\ref{sub:semant-expr-comp},

An interpretation~$H \oftype \interp{\sig}{Y}$ can interpret operations of $\sig$ as well as any template using them. For a well-typed template $\wftemplate{\tmplT}{\sig}$, we recursively define $\hsem[H]{\wftemplate{\tmplT}{\sig}} \oftype (\sem{\ctx} \times \sem{\tvars}) Y \to Y$ (often abbreviated to $\hsem[H]{\tmplT}$) as:
\begin{align*}
	\hsem[H]{\wftemplate{z_i(v)}{\sig}}(\eta; \zeta) &= \zeta_i(\sem{v}\eta) \\
	\hsem[H]{\wftemplate{\ifclause{v}{\tmplT_1}{\tmplT_2}}{\sig}}(\eta; \zeta) &=
	\begin{cases}
		\hsem[H]{\tmplT_1}(\eta; \zeta) & ; \text{ if }\sem{v}\eta = \semTrue  \\
		\hsem[H]{\tmplT_2}(\eta; \zeta) & ; \text{ if }\sem{v}\eta = \semFalse
	\end{cases} \\
	\hsem[H]{\wftemplate{\operation{v}{y.\tmplT}}{\sig}}(\eta; \zeta) &=
	H_\op	(\sem{v}\eta, \lambdafun{b}{\sem{\tyB_\op}}{\hsem[H]{\tmplT}(\eta, b; \zeta)}) \\
\end{align*}

\begin{lemma}\label{lem:handler_unroll}
	Take any sets $X, Y$, any interpretation $H \oftype \interp{\sig}{Y}$, any function $f \oftype X \to Y$, and define $\varphi = \lift[H]{f} \oftype \sem{\sig}X \to Y$. Then, the following diagram commutes
	\[
		\begin{tikzcd}
			(\sem{\ctx} \times \sem{\tvars}) (\sem{\sig} X)
			\arrow{r}{(\sem{\ctx} \times \sem{\tvars}) \varphi}
			\arrow{d}{\hsem[\freeInterp{X}{\sig}]{\tmplT}}
			& (\sem{\ctx} \times \sem{\tvars}) Y
			\arrow{d}{\hsem[H]{\tmplT}}
			\\
			\sem{\sig} X
			\arrow{r}{\varphi}
			& Y
		\end{tikzcd}
	\]

	In fact, all functions $\varphi \oftype \sem{\sig}X \to Y$ for which the diagram commutes are lifts of some function $f \oftype X \to Y$.
\end{lemma}

From Lemma~\ref{lem:handler_unroll}, it follows that the free interpretation~$\freeInterp{X}{\sig}$ of a template $\wftemplate{\tmplT}{\sig}$ yields a natural transformation $\alpha : (\sem{\ctx} \times \sem{\tvars}) \circ \sem{\sig} \Rightarrow \sem{\sig}$, given by:
\begin{equation*}
	\alpha_X
	= \hsem[\freeInterp{X}{\sig}]{\wftemplate{\tmplT}{\sig}}
\end{equation*}

%===============================================================================
\subsection{Semantics of theories}\label{sub:semantics-theories}
%===============================================================================

With denotation of templates in place, we can focus on effect theories. We remedy the absence of theories in the denotations $\sem{\tyA}$ and $\sem{\ctyC}$ by defining a family of relations $\rel_{\tyA}$ on $\sem{\tyA}$ and $\rel_{\ctyC}$ on $\sem{\ctyC}$. We also extend relations to contexts, and for $\ctx = x_1 \oftype A_1, \dots, x_n \oftype A_n$, we define the relation $\rel_\ctx$ on $\sem{\ctx}$ by
\[
	(a_1, \dots, a_n) \rel_\ctx (a_1', \dots, a_n') \iff a_1 \rel_{A_1} a_1' \land \dots \land a_n \rel_{A_n} a_n'	
\]
For value types, the relations are defined by:
\begin{itemize}
	\item 
		The relations on $\sem{\booltype}$ and $\sem{\unittype}$ are identities.
	\item 
		For a function type $\tyA \to \ctyC$ and $f, f' \oftype \sem{\tyA \to \ctyC}$, we define
		\[
		  f \rel_{\tyA \to \ctyC} f' \iff
		  (\fall{a,a' \in \sem{\tyA}} a \rel_\tyA a' \implies f(a)
		  \rel_{\ctyC} f'(a'))
		\]
	\item 
		For a handler type $\ctyC \hto \ctyD$ and $h, h' \oftype \sem{\ctyC \hto 	\ctyD}$, we define
		\[
		  h \rel_{\ctyC \hto \ctyD} h' \iff
		  (\fall{t,t' \in \sem{\ctyC}} t \rel_{\ctyC} t' \implies h(t)
		  \rel_{\ctyD} h'(t'))
		\]
\end{itemize}
For the computation type $\ctyC = \ctypebasic$, we define the relation~$\rel_{\ctyC}$ on the set $\sem{\ctypebasic}$ to be the smallest transitive and symmetric relation closed under the following rules:
\begin{itemize}
	\item If $a \rel_\tyA a'$ then $\inval(a) \rel_{\ctypebasic} \inval(a')$.
	\item For any operation $\op \oftype \tyA_\op \to \tyB_\op \in \sig$, if $a \rel_{\tyA_\op} a'$ and $\continvar \rel_{\tyB_\op \to \ctyC} \continvar'$ then also
	  \[ \inop(a; \continvar) \rel_{\ctyC} \inop(a'; \continvar'). \]
	\item For all equations $\teqjudgement{\tenv}{\tmplT_1 \rel \tmplT_2}$ in $\eqE$, where $\ctx = (x_i \oftype \tyA_i)_i$ and $\tvars = (z_j \oftype	\tyB_j \to \anytype)_j$, we say that if all $a_i \rel_{\tyA_i} a'_i$ and all $f_j \rel_{\tyB_j \to \ctyC} f'_j$ then
	  \[
		  \hsem[\freeInterp{\sem{A}}{\sig}]{\wftemplate{\tmplT_1}{\sig}}((a_i)_i, (f_j)_j)
		  \rel_{\ctyC}
		  \hsem[\freeInterp{\sem{A}}{\sig}]{\wftemplate{\tmplT_2}{\sig}}((a'_i)_i, (f'_j)_j)
	  \]
\end{itemize}
The last rule ensures that all equations in the effect theory~$\eqE$ are reflected in the relation~$\rel_{\ctyC}$, while all other definitions just propagate them structurally.

\begin{lemma}\label{lem:lift-preserves-rel}
	Let $\ctyC = \ctypebasic$ and $\ctyD = \ctype{\tyB}{\sig}{\eqE}$. If $g \rel_{\tyA \to \ctyD} g'$ then for every $t \rel_{\ctyC} t'$ it holds that 
	\[
		(\lift[\freeInterp{\sem{\tyA}}{\sig}]{g})(t) \rel_{\ctyD} (\lift[\freeInterp{\sem{\tyA}}{\sig}]{g'})(t')
	\]
\end{lemma}
\begin{proof}
	We shorten $\lift[\freeInterp{\sem{\tyA}}{\sig}]$ to $\lift[]$ and $\hsem[\freeInterp{\sem{\tyA}}{\sig}]{\wftemplate{T}{\sig}}$ to $\sem{T}$ for clarity. Since $\rel_{\ctyC}$ is defined to be the smallest relation closed under the given rules, we can proceed by induction on $t \rel_{\ctyC} t'$:
	\begin{enumerate}
		\item 
		First, consider $\inval(a) \rel_{\ctyC} \inval(a')$ where $a \rel_{\tyA} a'$. By definition of $\lift[]{}$ we obtain $(\lift{g})(\inval(a)) = g(a)$ and $(\lift{g'})(\inval(a')) = g'(a')$. Because $g \rel_{\tyA \to \ctyD} g'$ we know that the functions map related arguments to related images.
		\item 
		Next, take $\inop(a; \continvar) \rel_{\ctyC} \inop(a'; \continvar')$ for $\op \oftype \tyA_\op \to \tyB_\op \in \sig$ where $a \rel_{\tyA_\op} a'$ and $\continvar \rel_{\tyB_\op \to \ctyC} \continvar'$. We know that $(\lift{g})(\inop(a; \continvar)) = \inop{(a, \lift{g} \circ \continvar)}$ (and similarly for $g'$). By induction, $\lift{g} \circ \continvar$ is related to $\lift{g'} \circ \continvar'$, and so we obtain the desired result that $\inop{(a, \lift{g} \circ \continvar)} \rel_{\ctyD} \inop{(a', \lift{g'} \circ \continvar')}$.
		\item 
		Finally, we have the interesting case where $t \rel_{\ctyC} t'$ comes from an instantiation of an equation $\teqjudgement{\tenv}{\tmplT_1 \rel \tmplT_2} \in \eqE$, where we denote $\ctx = (x_i \oftype \tyA_i)_i$ and $\tvars = (z_j \oftype	\tyB_j \to \anytype)_j$. That means that there exist $a_i \rel_{\tyA_i} a'_i$ and $f_j \rel_{\tyB_j \to \ctyC} f'_j$ for which
		\[
			t = \sem{\tmplT_1}((a_i)_i, (f_j)_j) 
			\rel
			t' = \sem{\tmplT_2}((a'_i)_i, (f'_j)_j).
		\]
		From Lemma~\ref{lem:handler_unroll}, we have $(\lift{g})(\sem{\tmplT_1}((a_i)_i, (f_j)_j)) = \sem{\tmplT_1}((a_i)_i, (\lift{g} \circ f_j)_j)$ and similarly for $\tmplT_2$. By induction hypotheses we get $\lift{g} \circ f_j \rel \lift{g'} \circ f'_j$ as before, thus
		\[
			\sem{\tmplT_1}((a_i)_i, (\lift{g} \circ f_j)_j) 
			\rel
			\sem{\tmplT_2}((a'_i)_i, (\lift{g'} \circ f'_j)_j).
		\]
		concluding the proof.
	\end{enumerate}
\end{proof}

We can easily adapt Lemma~\ref{lem:adequacy} to consider equivalent computations, with $\rel_{\ctype{\booltype}{\nil}{\nil}}$ being the identity relation.

\begin{lemma}[Adequacy]\label{lem:better_adequacy}
	If $\jsem{}{c}{\ctype{\booltype}{\nil}{\nil}} \rel \inval{(\jsem{}{v}{\booltype})}$ then $c \leadsto^{*} \return{v}$.
\end{lemma}

Generalizing Corollary~\ref{cor:eq_adequacy} is trickier, however. Recall that the typing relation depends on the information about which handlers respect which effect theories. If we state that all handlers respect all theories, we can quickly produce a counterexample of a handler that maps equivalent computations into non-equivalent ones. In order for our relations to make sense, the reasoning logic~$\logic$ needs to respect effect theories.

\subsection{Soundness of a logic}

\begin{definition}\label{def:sound-logic}
	A logic~$\logic$ is \emph{sound} if $(\respects{\hcases}{\eqE}{\sig}{\ctyD})$ implies that for any
	\[
		\teqjudgement{(x_i \oftype \tyA_i)_i; (z_j \oftype	\tyB_j \to \anytype)_j}{T_1 \rel T_2} \in \eqE,
	\]
	any $\eta \rel_{\ctx} \eta'$, any $a_i \rel_{\tyA_i} a'_i$, and any $f_j \rel_{\tyB_j \to \ctyD} f'_j$, we have:
	\[
		\hsem[\sem{h}\eta]{\wftemplate{T_1}{\sig}}((a_i)_i; (f_j)_j)
		\rel_{\ctyD}
		\hsem[\sem{h}\eta']{\wftemplate{T_2}{\sig}}((a'_i)_i; (f'_j)_j)
	\]
\end{definition}

Not every logic is sound. For example, take any logic that contains
\[
	\respects[]{h_{\mathit{pickLeft}}}{\set{\eqref{comm}}}{\set{\mathit{choose}}}{(\ctype{\inttype}{\emptyset}{\emptyset})}
\]
where $h_{\mathit{pickLeft}}$ are the operation clauses of $\mathit{pickLeft}$ handler, and take $f_1 = f'_1 = \lambda\star. \inval(1)$ while $f_2 = f'_2 = \lambda\star. \inval(2)$. We end up with $\inval(1)$ and $\inval(2)$ which are not equivalent under $\rel_{\ctype{\inttype}{\emptyset}{\emptyset}}$.

\begin{proposition}\label{prop:sound-logic}
	If a logic~$\logic$ is sound, then for any $\eta \rel_{\ctx} \eta'$, we have:
	\begin{itemize}
		\item $\judgement{\ctx}{v}{\tyA}$ implies $\jsem{\ctx}{v}{\tyA}\eta \rel_\tyA \jsem{\ctx}{v}{\tyA}\eta'$.
		\item $\judgement{\ctx}{c}{\ctyC}$ implies $\jsem{\ctx}{c}{\ctyC}\eta \rel_{\ctyC} \jsem{\ctx}{c}{\ctyC}\eta'$.
	\end{itemize}	
\end{proposition}

\begin{proof}
	The majority of the proof proceeds with structural induction on the typing derivation, with the only non-trivial cases being sequencing and handlers.

	Recall the denotation of sequencing:
	\[
		\jsem{\ctx}{\dobind{x}{c_1}{c_2}}{\ctyD}\eta =
		\lift[\freeInterp{\sem{\tyA}}{\sig}]{(\lambdafun{a}{\sem{\tyA}}{\sem{c_2}{(\eta, a)}})}(\sem{c_1}\eta).
	\]
	By induction we know that $\sem{c_1}\eta \rel_{\ctyC} \sem{c_1}\eta'$ and $\sem{c_2}{(\eta, a)} \rel_{\ctyD} \sem{c_2}{(\eta', a')}$ for any $a \rel_\tyA a'$, which further implies that $(\lambdafun{a}{\sem{\tyA}}{\sem{c_2}{(\eta, a)}}) \rel_{\tyA \to \ctyD} (\lambdafun{a}{\sem{\tyA}}{\sem{c_2}{(\eta', a)}})$. We conclude by applying Lemma~\ref{lem:lift-preserves-rel}, which states that lifts of related functions are related.

	The other interesting case in the proof are the handlers. Recall the definition for handlers:
	\begin{equation*}
		\jsem{\ctx}{\handler \hcases}{\ctypebasic \hto \ctyD}\eta = \lift[\sem{\hcases}(\eta)]{(\lambdafun{a}{\sem{\tyA}}{\sem{c_r}(\eta, a)})}
	\end{equation*}
	where we define the meaning of handler cases $\jsem{\ctx}{\hcases}{\sig \casesto \ctyD}(\eta) \in \interp{\sig}{\sem{\ctyD}}$ as the family of functions:
	\[{\big \{
		\lambdafun{a}{\sem{\tyA_\op}}{
			\lambdafun{\kappa}{\sem{\tyB_\op \to \ctyD}}{
				\sem{c_\op}(\eta, a, \kappa)
			}}
		\big \}}_{\op{} \; \oftype \; \tyA_\op \to \tyB_\op \in \sig}
	\]
	We abbreviate the the denotation of the handler as $\semcases(\eta) := \lift[\sem{\hcases}(\eta)]{(\lambdafun{a}{\sem{\tyA}}{\sem{c_r}(\eta, a)})}$. We must prove that $\semcases(\eta)$ satisfies the requirements for the relation on handler types, i.e.\  for any $t \rel_{\ctyC} t'$ it must hold that $\semcases(\eta)(t) \rel_{\ctyD} \semcases(\eta')(t')$. 
	
	Because $\rel_{\ctyC}$ is the smallest relation closed under the given rules, we proceed by induction on $t \rel_{\ctyC} t'$. The cases where the relation is structural are largely the same as in proof of Lemma~\ref{lem:lift-preserves-rel}. The interesting case is the case of equivalence due to equations.

	Suppose that the relation $t \rel_{\ctyC} t'$ arises from the equation $\teqjudgement{\tenv}{\tmplT_1 \rel \tmplT_2} \in \eqE$, where we denote $\ctx = (x_i \oftype \tyA_i)_i$ and $\tvars = (z_j \oftype	\tyB_j \to \anytype)_j$. That means that there exist $a_i \rel_{\tyA_i} a'_i$ and $f_j \rel_{\tyB_j \to \ctyC} f'_j$ for which
		\[
			t = \sem{\tmplT_1}((a_i)_i, (f_j)_j) 
			\rel
			t' = \sem{\tmplT_2}((a'_i)_i, (f'_j)_j).
		\]

	By Lemma~\ref{lem:handler_unroll} we have the equality: 
	\begin{align*}
	\semcases(\eta)(\sem{\tmplT_1}_{\ctyC}((a_i)_i,(f_j)_j)) = \hsem[\sem{h}\eta]{\tmplT_1}\big((a_i)_i;(\semcases(\eta) \circ f_j)_j\big)
	\end{align*}
	and similarly for $\tmplT_2$.

	By induction hypotheses we get $\semcases(\eta) \circ f_j \rel_{\tyB_j \to \ctyD} \semcases(\eta') \circ f_j'$, so by the assumption on logic soundness, we have by Definition~\ref{def:sound-logic}:
	\[
		\hsem[\sem{h}\eta]{\wftemplate{T_1}{\sig}}((a_i)_i; (\semcases(\eta) \circ f_j)_j)
		\rel_{\ctyD}
		\hsem[\sem{h}\eta']{\wftemplate{T_2}{\sig}}((a'_i)_i; (\semcases(\eta') \circ f'_j)_j)
	\]
	which concludes our proof.
\end{proof}

\begin{theorem}\label{theorem:rel-ctxequiv}
	Assume a logic~$\logic$ is sound. Then, we have:
	\begin{itemize}
		\item if $\jsem{\ctx}{c}{\ctyC}\eta \rel_{\ctyC} \jsem{\ctx}{c'}{\ctyC}\eta'$ holds for any $\eta \rel_{\ctx} \eta'$, then $\ctxequiv{c}{c'}{\ctyC}$;
		\item if $\jsem{\ctx}{v}{\tyA}\eta \rel_{\tyA} \jsem{\ctx}{v'}{\tyA}\eta'$ holds for any $\eta \rel_{\ctx} \eta'$, then $\ctxequiv{v}{v'}{\tyA}$.
	\end{itemize}
\end{theorem}
\begin{proof}
	The proof is an adaptation of Corollary~\ref{cor:eq_adequacy}. First, assume that $\CEinstantiate{\CEvar}{c} \leadsto^* \return{\true}$. By Proposition~\ref{prop:sound-logic}, we have $\sem{\CEinstantiate{\CEvar}{c}} \rel \sem{\CEinstantiate{\CEvar}{c}}$, and by Proposition~\ref{prop:soundness-operational}, we have $\sem{\CEinstantiate{\CEvar}{c}} = \inval(\semTrue)$. Using a inductive proof similar to one for Proposition~\ref{prop:sound-logic}, we can show that the denotational semantics is structural with respect to equivalence and so $\sem{\CEinstantiate{\CEvar}{c}} \rel \sem{\CEinstantiate{\CEvar}{c'}}$. We conclude by applying Lemma~\ref{lem:better_adequacy}. The proof for values is identical.
\end{proof}

\subsection{Examples of sound logics}

The logics~$\emptylogic$, $\freelogic$ and $\fulllogic$, presented in Section~\ref{sec:logics} are indeed sound. Next, consider the equational logic, given in~\ref{sub:equational-logic}. In addition to soundness, we see that all terms shown to be equivalent under the logic are also equivalent with respect to the semantics.

\begin{theorem}\label{theorem:equational-logic-sound}
	The equational logic~$\eqlogic$ is sound and we have
	\begin{itemize}
		\item If $\ljudgement{\ctx}{}{v \lequal_{\tyA} v'}$ then $\jsem{\ctx}{v}{\tyA}\eta \rel_\tyA \jsem{\ctx}{v'}{\tyA}\eta'$.
		\item If $\ljudgement{\ctx}{}{c \lequal_{\ctyC} c'}$ then $\jsem{\ctx}{c}{\ctyC}\eta \rel_{\ctyC} \jsem{\ctx}{c'}{\ctyC}\eta'$.
	\end{itemize}
\end{theorem}

\begin{proof}
	The proof proceeds by induction on the derivation of typing and equality judgements. Note that we cannot simply apply Proposition~\ref{prop:sound-logic} because of the mutual dependence between typing judgements and equality judgements, however we can follow the structure of its proof. Most cases are immediate apart from inheritance from the effect theory, which amounts exactly to the definition of $\rel_{\ctypebasic}$, and handler validation, which follows from the assumption that~$\eqlogic$ is sound.
\end{proof}

Composing Theorems~\ref{theorem:rel-ctxequiv} and~\ref{theorem:equational-logic-sound} we then get:

\begin{corollary}\label{cor:lequal-to-ctxequiv}
	If $\eqjudg{v_1}{v_2}{\tyA}$, then $\ctxequiv{v_1}{v_2}{\tyA}$.
	If $\eqjudg{c_1}{c_2}{\ctyC}$, then $\ctxequiv{c_1}{c_2}{\ctyC}$.
\end{corollary}

In order to show that~$\predlogic$ is sound, we need to define the denotation of formulae. We interpret each formula $\varphi$ in a context $\ctx$ as a relation on $\sem{\ctx}$, defined by:

\begin{align*}
	\eta \sem{v_1 \lequal_\tyA v_2} \eta' & \iff \jsem{\ctx}{v_1}{\tyA}{\eta} \rel_{\ctyC} \jsem{\ctx}{v_2}{\tyA}{\eta'} \\
	\eta \sem{c_1 \lequal_{\ctyC} c_2} \eta' & \iff \jsem{\ctx}{c_1}{{\ctyC}}{\eta} \rel_{\ctyC} \jsem{\ctx}{c_2}{{\ctyC}}{\eta'} \\
	\eta \sem{\ltrue} \eta'                        &\iff \eta \rel_\ctx \eta'                              \\
	\eta \sem{\lfalse} \eta'                       &\iff \lfalse                             \\
	\eta \sem{\varphi_1 \land \varphi_2} \eta'           &\iff (\eta \sem{\varphi_1} \eta') \land (\eta \sem{\varphi_2} \eta') \\
	\eta \sem{\varphi_1 \lor \varphi_2} \eta'            &\iff (\eta \sem{\varphi_1} \eta') \lor (\eta \sem{\varphi_2} \eta') \\
	\eta \sem{\varphi \limplies \psi} \eta'       &\iff (\eta \sem{\varphi} \eta') \limplies (\eta \sem{\psi} \eta') \\
	\eta \sem{\lall{x \oftype \tyA} \varphi} \eta'    &\iff \lall{a \rel_\tyA a'} (\eta, a) \sem{\varphi} (\eta', a') \\
	\eta \sem{\lexists{x \oftype \tyA} \varphi} \eta'    &\iff \lexists{a \rel_\tyA a'} (\eta, a) \sem{\varphi} (\eta', a')
\end{align*}

Like in proof of Theorem~\ref{theorem:equational-logic-sound}, we proceed by an induction on the judgement derivation and show soundness of~$\predlogic$. Most cases are identical to~$\eqlogic$ or follow from the defining properties of logical connectives and quantifiers. Soundness of induction follows straight from the inductive structure of computation types~\cite{DBLP:conf/lics/PlotkinP08, DBLP:phd/ethos/Pretnar10, DBLP:journals/corr/BauerP13}.

\begin{theorem}
	The equational logic~$\predlogic$ is sound and if $\ljudgement{\ctx}{\psi_1, \dots, \psi_n}{\varphi}$ holds, then for any $\eta$ and $\eta'$ such that $\eta \sem{\psi_1 \land \cdots \land \psi_n} \eta'$, we have $\eta \sem{\varphi} \eta'$.
\end{theorem}

\section{Conclusion}\label{sec:conclusion}
% !TeX root = main.tex
% chktex-file 1
% chktex-file 46

\subsection{Related work}

\paragraph{Equational reasoning about monadic effects}
Simple equational reasoning is one of the better properties that pure functional programs enjoy. As suggested in the seminal paper on monads~\cite{DBLP:journals/iandc/Moggi91}, this approach can be extended to functional programs that use monadic effects. Further examples of such reasoning can be found in~\cite{DBLP:conf/icfp/GibbonsH11}. In addition to equations over pure programs, reasoning about effectful programs employs equations describing propagation of operation calls, monadic axioms, and equations describing primitive effectful operations. All three kinds of equations are fully supported in our approach (the first is an axiom, the second is derivable with induction, and the third corresponds to inheritance from effect theories), allowing us to use the same techniques with the additional flexibility of locally varying the effect theory.

\paragraph{Algebraic effects and dependent types}
Another approach to reason about effectful programs through richer type annotations is to employ dependent types~\cite{DBLP:conf/icfp/Brady13,Ahman:PhDThesis, DBLP:journals/pacmpl/Ahman18}. In fact,~\cite{DBLP:journals/pacmpl/Ahman18} is the only research work on handlers besides the original one that considers non-trivial effect theories. Its aim is similar to ours, though the implementation differs. The biggest difference is in the representation of handlers: dependent typing allows the operation clauses to be encoded in types, so any algebra for the theory has a matching type, whereas in our approach, only the free ones do.

\paragraph{Program optimisations}
Effect theories allow us to rewrite programs into equivalent but more efficient ones. For example, idempotence of a non-deterministic choice operation allows us to skip repeated computations, while commutativity allows us to change the order of evaluation, enabling further optimizations. A survey of such transformations can be found in~\cite{DBLP:conf/popl/KammarP12}. Even though the development was done in the context of a global effect theory, its results are easily adapted to our work. Though we understand that our work is far from reaching practical type-driven optimizations, we hope that it might prove to be useful in further development.

\subsection{Future work}\label{sub:future-work}

\paragraph{Language extensions}
For the sake of a clearer presentation, we have limited our working calculus to the smallest possible fragment that already exhibits the novel features of the type system. Extending the language with additional value types such as sums or products is simple, one just needs to take care to extend the template syntax with additional value destructors. Similarly, one can extend the language with recursive functions, though in this case one must switch the denotational semantics from the category of sets to one of domains and adapt the logics to divergence~\cite{DBLP:journals/corr/BauerP13}.

\paragraph{Subtyping}
A sensible extension is the addition of structural subtyping as in~\cite{DBLP:journals/corr/BauerP13, DBLP:conf/esop/SalehKPS18}. We only need to modify the rules for computation types since the addition of equations only affects computations. The simplest version is to allow $\ctype{\tyA}{\sig}{\eqE} \subtypeof \ctype{\tyA'}{\sig'}{\eqE'}$ whenever $\tyA \subtypeof \tyA'$, every operation in $\sig$ appears in $\sig'$ with a greater type, and when the theory $\eqE'$ entails all the equations in $\eqE$ (we may always consider more computations to be equivalent). The exact logic that allows such reasoning can be considered in future work. A simpler variant is to check that every equation in $\eqE$ appears in $\eqE'$ as well, though this approach is limited due to non-canonicity of the set of equations.

\paragraph{Formalization}
Similar to our previous work~\cite{DBLP:journals/corr/BauerP13, DBLP:journals/jfp/KammarP17, DBLP:journals/pacmpl/0002KLP17, DBLP:conf/esop/SalehKPS18}, we wish to mechanize our formalization in a proof assistant. From past experience, we expect the proof of Theorem~\ref{theorem:safety} to proceed smoothly, but expect bigger problems with proofs that depend on denotational semantics (e.g.\ Corollary~\ref{cor:lequal-to-ctxequiv}). For that reason, we plan to look at purely syntactical treatments of contextual equivalence~\cite{DBLP:conf/cpp/McLaughlinMS18} that are more amenable to mechanization.

\paragraph{Practical implementation}
We plan on implementing the proposed system in the Eff programming language~\cite{DBLP:journals/jlp/BauerP15}. The user would annotate computation types with the desired equations, and the system would ensure they are respected. There are multiple interesting implementations to consider, such as dispatching the proofs to an SMT solver~\cite{DBLP:conf/tacas/MouraB08}, generating proof assistant templates that a user must complete with proofs, or using a QuickCheck~\cite{DBLP:conf/icfp/ClaessenH00} like tool, used to detect errors by running the given handler on random examples generated from equations and comparing the results on both sides.

\paragraph{Polymorphism}
Another aspect to consider in the practical implementation is the interaction with the currently implemented polymorphic core language~\cite{DBLP:conf/esop/SalehKPS18} where one may consider functions that are polymorphic both in the type and signature of an operation, for example
\[
    \mathit{map} \oftype
    \forall \alpha, \beta, \sigma. (\alpha \to \ctype{\beta}{\sigma}{}) \to \ctype{(\alpha~\code{list} \to \ctype{\beta~\code{list}}{\sigma}{})}{\nil}{}
\]
Since we want such functions to preserve the equivalences between computations, it would be natural to consider polymorphism in all the components of a computation type, which would allow us to assign a type such as
\[
    \mathit{map} \oftype
    \forall \alpha, \beta, \sigma, \varepsilon. (\alpha \to \ctype{\beta}{\sigma}{\varepsilon}) \to \ctype{(\alpha~\code{list} \to \ctype{\beta~\code{list}}{\sigma}{\varepsilon})}{\nil}{}
\]

\paragraph{Additional examples}
All of the examples currently presented are kept minimal for clarity. By implementing the system we plan on producing larger and more complex examples, showcasing the usefulness of our proposed system. Extending the system with some form of subtyping will also make examples more composable.

%- - - - - - - - - - - - - - - - - - - - - - - - - - - - - - - - - - - - - - - -

\subsection*{Acknowledgements}

We would like to thank Danel Ahman, Bob Atkey, Andrej Bauer, Jeremy Gibbons, Ohad Kammar, Oleg Kiselyov, Gordon Plotkin, Alex Simpson, Tom Schrijvers, Nicolas Wu, all other participants of Dagstuhl seminars 16112~\cite{DBLP:journals/dagstuhl-reports/Bauer0PY16} and 18172~\cite{DBLP:journals/dagstuhl-reports/ChandrasekaranL18}, and the anonymous referees for all their helpful insights and suggestions.

\bibliographystyle{plain}
\bibliography{bibliography}

\begin{thebibliography}{10}

\bibitem{Ahman:PhDThesis}
Danel Ahman.
\newblock {\em Fibred Computational Effects}.
\newblock PhD thesis, School of Informatics, University of Edinburgh, 2017.

\bibitem{DBLP:journals/pacmpl/Ahman18}
Danel Ahman.
\newblock Handling fibred algebraic effects.
\newblock {\em {PACMPL}}, 2({POPL}):7:1--7:29, 2018.

\bibitem{DBLP:journals/dagstuhl-reports/Bauer0PY16}
Andrej Bauer, Martin Hofmann, Matija Pretnar, and Jeremy Yallop.
\newblock From theory to practice of algebraic effects and handlers (dagstuhl
  seminar 16112).
\newblock {\em Dagstuhl Reports}, 6(3):44--58, 2016.

\bibitem{DBLP:journals/corr/BauerP13}
Andrej Bauer and Matija Pretnar.
\newblock An effect system for algebraic effects and handlers.
\newblock {\em Logical Methods in Computer Science}, 10(4), 2014.

\bibitem{DBLP:journals/jlp/BauerP15}
Andrej Bauer and Matija Pretnar.
\newblock Programming with algebraic effects and handlers.
\newblock {\em J. Log. Algebr. Meth. Program.}, 84(1):108--123, 2015.

\bibitem{DBLP:journals/pacmpl/BiernackiPPS18}
Dariusz Biernacki, Maciej Pir{\'{o}}g, Piotr Polesiuk, and Filip Sieczkowski.
\newblock Handle with care: relational interpretation of algebraic effects and
  handlers.
\newblock {\em {PACMPL}}, 2({POPL}):8:1--8:30, 2018.

\bibitem{DBLP:conf/icfp/Brady13}
Edwin Brady.
\newblock Programming and reasoning with algebraic effects and dependent types.
\newblock In Greg Morrisett and Tarmo Uustalu, editors, {\em {ACM} {SIGPLAN}
  International Conference on Functional Programming, ICFP'13, Boston, MA,
  {USA} - September 25 - 27, 2013}, pages 133--144. {ACM}, 2013.

\bibitem{DBLP:journals/dagstuhl-reports/ChandrasekaranL18}
Sivaramakrishnan~Krishnamoorthy Chandrasekaran, Daan Leijen, Matija Pretnar,
  and Tom Schrijvers.
\newblock Algebraic effect handlers go mainstream (dagstuhl seminar 18172).
\newblock {\em Dagstuhl Reports}, 8(4):104--125, 2018.

\bibitem{DBLP:conf/icfp/ClaessenH00}
Koen Claessen and John Hughes.
\newblock Quickcheck: a lightweight tool for random testing of haskell
  programs.
\newblock In Martin Odersky and Philip Wadler, editors, {\em Proceedings of the
  Fifth {ACM} {SIGPLAN} International Conference on Functional Programming
  {(ICFP} '00), Montreal, Canada, September 18-21, 2000.}, pages 268--279.
  {ACM}, 2000.

\bibitem{DBLP:conf/tacas/MouraB08}
Leonardo~Mendon{\c{c}}a de~Moura and Nikolaj Bj{\o}rner.
\newblock {Z3:} an efficient {SMT} solver.
\newblock In {\em {TACAS}}, volume 4963 of {\em Lecture Notes in Computer
  Science}, pages 337--340. Springer, 2008.

\bibitem{DBLP:journals/pacmpl/0002KLP17}
Yannick Forster, Ohad Kammar, Sam Lindley, and Matija Pretnar.
\newblock On the expressive power of user-defined effects: effect handlers,
  monadic reflection, delimited control.
\newblock {\em {PACMPL}}, 1({ICFP}):13:1--13:29, 2017.

\bibitem{DBLP:conf/icfp/GibbonsH11}
Jeremy Gibbons and Ralf Hinze.
\newblock Just do it: simple monadic equational reasoning.
\newblock In Manuel M.~T. Chakravarty, Zhenjiang Hu, and Olivier Danvy,
  editors, {\em Proceeding of the 16th {ACM} {SIGPLAN} international conference
  on Functional Programming, {ICFP} 2011, Tokyo, Japan, September 19-21, 2011},
  pages 2--14. {ACM}, 2011.

\bibitem{DBLP:conf/icfp/HillerstromL16}
Daniel Hillerstr{\"{o}}m and Sam Lindley.
\newblock Liberating effects with rows and handlers.
\newblock In James Chapman and Wouter Swierstra, editors, {\em Proceedings of
  the 1st International Workshop on Type-Driven Development, TyDe@ICFP 2016,
  Nara, Japan, September 18, 2016}, pages 15--27. {ACM}, 2016.

\bibitem{DBLP:conf/icfp/KammarLO13}
Ohad Kammar, Sam Lindley, and Nicolas Oury.
\newblock Handlers in action.
\newblock In Greg Morrisett and Tarmo Uustalu, editors, {\em {ACM} {SIGPLAN}
  International Conference on Functional Programming, ICFP'13, Boston, MA,
  {USA} - September 25 - 27, 2013}, pages 145--158. {ACM}, 2013.

\bibitem{DBLP:conf/popl/KammarP12}
Ohad Kammar and Gordon~D. Plotkin.
\newblock Algebraic foundations for effect-dependent optimisations.
\newblock In John Field and Michael Hicks, editors, {\em Proceedings of the
  39th {ACM} {SIGPLAN-SIGACT} Symposium on Principles of Programming Languages,
  {POPL} 2012, Philadelphia, Pennsylvania, USA, January 22-28, 2012}, pages
  349--360. {ACM}, 2012.

\bibitem{DBLP:journals/jfp/KammarP17}
Ohad Kammar and Matija Pretnar.
\newblock No value restriction is needed for algebraic effects and handlers.
\newblock {\em J. Funct. Program.}, 27:e7, 2017.

\bibitem{DBLP:conf/popl/Leijen17}
Daan Leijen.
\newblock Type directed compilation of row-typed algebraic effects.
\newblock In Giuseppe Castagna and Andrew~D. Gordon, editors, {\em Proceedings
  of the 44th {ACM} {SIGPLAN} Symposium on Principles of Programming Languages,
  {POPL} 2017, Paris, France, January 18-20, 2017}, pages 486--499. {ACM},
  2017.

\bibitem{DBLP:journals/iandc/LevyPT03}
Paul~Blain Levy, John Power, and Hayo Thielecke.
\newblock Modelling environments in call-by-value programming languages.
\newblock {\em Inf. Comput.}, 185(2):182--210, 2003.

\bibitem{DBLP:conf/cpp/McLaughlinMS18}
Craig McLaughlin, James McKinna, and Ian Stark.
\newblock Triangulating context lemmas.
\newblock In June Andronick and Amy~P. Felty, editors, {\em Proceedings of the
  7th {ACM} {SIGPLAN} International Conference on Certified Programs and
  Proofs, {CPP} 2018, Los Angeles, CA, USA, January 8-9, 2018}, pages 102--114.
  {ACM}, 2018.

\bibitem{DBLP:journals/iandc/Moggi91}
Eugenio Moggi.
\newblock Notions of computation and monads.
\newblock {\em Inf. Comput.}, 93(1):55--92, 1991.

\bibitem{DBLP:conf/fossacs/PlotkinP01}
Gordon~D. Plotkin and John Power.
\newblock Adequacy for algebraic effects.
\newblock In Furio Honsell and Marino Miculan, editors, {\em Foundations of
  Software Science and Computation Structures, 4th International Conference,
  {FOSSACS} 2001 Held as Part of the Joint European Conferences on Theory and
  Practice of Software, {ETAPS} 2001 Genova, Italy, April 2-6, 2001,
  Proceedings}, volume 2030 of {\em Lecture Notes in Computer Science}, pages
  1--24. Springer, 2001.

\bibitem{DBLP:journals/acs/PlotkinP03}
Gordon~D. Plotkin and John Power.
\newblock Algebraic operations and generic effects.
\newblock {\em Applied Categorical Structures}, 11(1):69--94, 2003.

\bibitem{DBLP:conf/lics/PlotkinP08}
Gordon~D. Plotkin and Matija Pretnar.
\newblock A logic for algebraic effects.
\newblock In {\em Proceedings of the Twenty-Third Annual {IEEE} Symposium on
  Logic in Computer Science, {LICS} 2008, 24-27 June 2008, Pittsburgh, PA,
  {USA}}, pages 118--129. {IEEE} Computer Society, 2008.

\bibitem{DBLP:conf/esop/PlotkinP09}
Gordon~D. Plotkin and Matija Pretnar.
\newblock Handlers of algebraic effects.
\newblock In Giuseppe Castagna, editor, {\em Programming Languages and Systems,
  18th European Symposium on Programming, {ESOP} 2009, Held as Part of the
  Joint European Conferences on Theory and Practice of Software, {ETAPS} 2009,
  York, UK, March 22-29, 2009. Proceedings}, volume 5502 of {\em Lecture Notes
  in Computer Science}, pages 80--94. Springer, 2009.

\bibitem{DBLP:journals/corr/PlotkinP13}
Gordon~D. Plotkin and Matija Pretnar.
\newblock Handling algebraic effects.
\newblock {\em Logical Methods in Computer Science}, 9(4), 2013.

\bibitem{DBLP:phd/ethos/Pretnar10}
Matija Pretnar.
\newblock {\em Logic and handling of algebraic effects}.
\newblock PhD thesis, University of Edinburgh, {UK}, 2010.

\bibitem{DBLP:journals/entcs/Pretnar15}
Matija Pretnar.
\newblock An introduction to algebraic effects and handlers. {I}nvited tutorial
  paper.
\newblock {\em Electr. Notes Theor. Comput. Sci.}, 319:19--35, 2015.

\bibitem{DBLP:conf/esop/SalehKPS18}
Amr~Hany Saleh, Georgios Karachalias, Matija Pretnar, and Tom Schrijvers.
\newblock Explicit effect subtyping.
\newblock In Amal Ahmed, editor, {\em Programming Languages and Systems - 27th
  European Symposium on Programming, {ESOP} 2018, Held as Part of the European
  Joint Conferences on Theory and Practice of Software, {ETAPS} 2018,
  Thessaloniki, Greece, April 14-20, 2018, Proceedings}, volume 10801 of {\em
  Lecture Notes in Computer Science}, pages 327--354. Springer, 2018.

\end{thebibliography}

\end{document}